%% file: main.tex
\documentclass[journal=tches]{iacrtrans}
\input{setup.tex}

\usepackage{acro}
\input{acronyms.tex}

\author{Gustavo Banegas\inst{1} \and YoungBeom Kim\inst{2} \and Seog Chung Seo\inst{2} \\ \and Christine van Vredendaal\inst{3}}
\institute{
   LIX, CNRS, Inria, École Polytechnique, Institut Polytechnique de Paris, France
 \email{gustavo.souza-banegas@inria.fr}
 \and
 Kookmin University, Seoul, Republic of Korea, \email{{darania, scseo}@kookmin.ac.kr}
 \and
 NXP Semiconductors, The Netherlands, \email{christine.cloostermans@nxp.com}}

\title{Low-Stack HAETAE for Memory-Constrained Microcontrollers}

\begin{document}

\maketitle


  \makeatletter%
\begingroup%
 \makeatletter%
 \def\@thefnmark{$*$}\relax%
 \@footnotetext{\relax%
   Author list in alphabetical order; see \url{https://ams.org/profession/leaders/CultureStatement04.pdf}.
   This work was supported by the HYPERFORM consortium, funded by France through Bpifrance, and by the France
   2030 program under grant agreement ANR-22-PETQ-0008 PQ-TLS.
\def\ymdtoday{\leavevmode\hbox{\the\year-\twodigits\month-\twodigits\day}}\def\twodigits#1{\ifnum#1<10 0\fi\the#1}%
  Date of this document: \ymdtoday.%
 }%
\endgroup
\begin{abstract}
We present a low-stack implementation of the module-lattice signature scheme
\hae{}, targeting microcontrollers with
\SI{8}{\kilo\byte}--\SI{16}{\kilo\byte} of available SRAM.
On such devices, peak stack usage is often the binding constraint, and
\hae{}'s hyperball-based sampler, large transient polynomial
vectors, and variable-length signature payloads (hint and high-bits
arrays) pose a particular challenge.
To address this we introduce
(i)~\emph{Rejection-aware pass decomposition},
which isolates encoding to the post-acceptance path;
(ii)~\emph{Component-level early rejection},
which short-circuits the response computation when a partial norm
already exceeds the bound; and
(iii)~\emph{Reverse-order streaming entropy coding} using range
Asymmetric Numeral Systems (rANS), which eliminates full hint and
high-bits staging buffers.
Combined with streamed matrix generation, a two-pass hyperball
sampler with streaming Gaussian backend, and row-streamed
verification, these techniques bring Signing stack from
\SI{71}{\kilo\byte}--\SI{141}{\kilo\byte} in the reference
implementation down to
\SI{5.8}{\kilo\byte}--\SI{6.0}{\kilo\byte},
key generation to
\SI{4.7}{\kilo\byte}--\SI{5.7}{\kilo\byte}, and verification to
\SI{4.7}{\kilo\byte}--\SI{4.8}{\kilo\byte} across all three security levels.
Our pure~C implementation covers all three security levels
(\hae{}-2/3/5), whose optimization paths differ due to the
public-key domain ($d{>}0$ vs.\ $d{=}0$) and rejection structure.
We implement our optimization on a Nucleo-L4R5ZI and compare to the reference \texttt{pqm4} (for \hae{}-2 and -3) and a recently published memory-optimized
implementation (targeting \hae{}-5 only).
We reduce \hae{}-2, -3, and -5
stack by respectively 75, 86 and 8\,\% for
key generation, 92, 95 and 24\,\% for signature generation, and 85, 91 and 22\,\% for verification.
Depending on the parameter set, this impacts performance by at most a factor 1.8 and 3.4 for key and signature generation respectively, while even offering a performance improvement up to 18\,\% for verification.
Verification at all security levels fits within \SI{8}{\kilo\byte}
of RAM (signature buffer + stack) and is $2.34$--$3.34\times$ faster
than ML-DSA m4fstack at each comparable security level.
We additionally validate portability under RIOT-OS on ARM Cortex-M4
and RISC-V targets.
\end{abstract}

\input{introduction}

\input{background.tex}

\input{theoretical_eval.tex}

\input{techniques.tex}

\input{evaluation.tex}

\input{conclusion.tex}


%
\bibliographystyle{alpha}
\bibliography{biblio}

\appendix
\input{appendix.tex}

\end{document}

%% file: setup.tex
\usepackage{svg}
\usepackage{hyperref}
\usepackage{pifont}
\usepackage{algorithm}
\usepackage{algorithmicx}
\usepackage{algpseudocode}

\usepackage{listings}
\usepackage{graphicx}
\usepackage{subcaption}  
\usepackage{array}
\usepackage{siunitx}
\newcolumntype{Y}{>{\raggedright\arraybackslash}X}
\usepackage{arydshln}
\usepackage{booktabs}
\usepackage{multirow}
\usepackage{tabularx}
\usepackage{makecell}

\usepackage[autostyle=false, style=english]{csquotes}
\MakeOuterQuote{"}
\usepackage{tikz}
\usetikzlibrary{positioning}
\usetikzlibrary{shapes.geometric, arrows}
\usetikzlibrary{arrows.meta,calc}
\usepackage{arydshln} 

\tikzstyle{process} = [rectangle, minimum width=2cm, minimum height=0.7cm, text centered, draw=black] 
\tikzstyle{arrow} = [thick,->,>=stealth]

\newcommand{\stepn}[1]{\tikz[baseline=(c.base)]{\node[circle,draw,inner sep=0pt,minimum size=0.9em,line width=0.4pt](c){\scriptsize #1};}}

%% file: acronyms.tex
\newcommand{\pk}{\mathit{pk}}
\newcommand{\sk}{\mathit{sk}}
\newcommand{\spo}{{\bf s}}
\newcommand{\epo}{{\bf e}}

\newcommand{\Apo}{{\bf A}}
\newcommand{\apo}{{\bf a}}
\newcommand{\bpo}{{\bf b}}

\newcommand{\ntt}{\textsf{NTT}}
\newcommand{\invntt}{\ntt^{-1}}

\newcommand{\hae}{HAETAE}

\DeclareAcronym{nist}{short=NIST, long={National Institute of Standards and Technology}}

\DeclareAcronym{ntt}{
    short=NTT,
    long=Number Theoretic Transform
}
\DeclareAcronym{ecdsa}{
    short=ECDSA,
    long=Elliptic Curve Digital Signature Algorithm
}

\DeclareAcronym{kats}{short=KATs, long=Known-Answer Tests}

\DeclareAcronym{pqm4}{short=PQM4, long=Library that is a collection of post-quantum cryptographic algorithms for the ARM Cortex-M4}

\newcommand{\InitStream}{\texttt{InitStream}}

\newcommand{\SampleGauss}{\texttt{SampleGauss}}
\newcommand{\InvSqrt}{\texttt{InvSqrt}}

\newcommand{\round}{\texttt{Round}}

\newcommand{\tikzoffset}{2pt}

\newcommand\tikzmark[2]{%
	\tikz[remember picture,baseline] \node[inner sep=2pt,outer sep=0pt] (#1){#2};%
}

\newif\iftodoenabled
\todoenabledtrue

%% file: introduction.tex
\section{Introduction}
\label{sec:introduction}

Embedded systems rely on public-key authentication schemes for firmware updates,
secure boot, device attestation, and device provisioning.
With the transition to post-quantum cryptography, signature schemes based
on module lattices are leading candidates for standardization.

On low-end ARM Cortex-M and RISC-V microcontrollers, the stack
available to an application may be only
\SI{8}{\kilo\byte}--\SI{16}{\kilo\byte} once the real-time operating system (RTOS), network stack,
and I/O drivers are loaded.
For this class of targets, peak stack usage is often the binding
constraint rather than code size or cycle count.

\hae{}~\cite{NISTPQC-ADD-R1:HAETAE23,cheon2024haetae}, a
module-lattice signature scheme presented at CHES'24, achieves
shorter signatures and keys than ML-DSA~\cite{FIPS204} by sampling
ephemeral vectors from a hyperball rather than a hypercube.
This compactness reduces communication and persistent storage costs,
making \hae{} attractive for constrained deployments.
For its attractive qualities, it has been identified as a winner in the Korean domestic post-quantum cryptography competition, KpqC~\cite{KPQC_Exhibit}. 

\paragraph{Motivation.}
Although \hae{} offers smaller public keys and signatures than ML-DSA
at comparable security levels (Table~\ref{tab:haetae-params}), the
reference implementation requires \SI{71}{\kilo\byte} to
\SI{141}{\kilo\byte} of peak stack during the signing operation, depending on the security
level, far exceeding the RAM budget of typical embedded targets.
The root cause is that multiple large polynomial vectors used in
sampling, challenge derivation, and hint computation are kept live
simultaneously.
The hyperball sampler amplifies the problem: because the scaling
factor depends on the squared norm of the full Gaussian sample, a
one-pass implementation must buffer the entire sample before producing
any output.

This issue extends beyond signing.
In modern embedded systems, key generation may also be performed at
runtime for ephemeral credentials or key rotation, and verification
must complete within tight stack budgets on sensor nodes that only
verify signatures.
A practical low-memory implementation should therefore address
key generation, signing, and verification alike.

\paragraph{Related work.}
Bos, Renes, and Sprenkels~\cite{DBLP:conf/africacrypt/BosRS22}
showed that ML-DSA can fit within a few kilobytes of stack by trading
recomputation for peak memory, streaming seed-derived objects, and
reusing lifetime-disjoint buffers; similar strategies have been applied
to FrodoKEM~\cite{bos2023enabling}.
For \hae{}, however, only streaming matrix generation and sparse
challenge multiplication carry over directly; the hyperball sampler,
variable-length rANS encoding, and hint computation remain unaddressed
by~\cite{DBLP:conf/africacrypt/BosRS22}.
To the best of our knowledge, \cite{cryptoeprint:2026/442} is the
state-of-the-art memory-optimized \hae{} implementation on Cortex-M4.
Following~\cite{DBLP:conf/africacrypt/BosRS22}, it achieves
\SI{5212}{\byte}/\SI{8092}{\byte}/\SI{6220}{\byte} (key
generation/signing/verification) for \hae{}-5, approaching the
idealized baselines of Section~\ref{sec:theory}.
However, it targets \hae{}-5 only, leaving \hae{}-2/3
($d{>}0$, different public-key domain and rejection structure)
unaddressed, and its verification stack (\SI{6220}{\byte}) plus the
signature (\SI{2948}{\byte}) totals \SI{9168}{\byte}, exceeding the
\SI{8}{\kilo\byte} budget of many constrained devices.
Fitting within this budget requires optimizations beyond general
streaming, targeting \hae{}-specific components that prior work
leaves unexplored.

\subsection{Contributions}

\paragraph{Support for all \hae{} security levels.}
We address both gaps identified above: our implementation supports
all three security levels (\hae{}-2/3/5), with separate optimization
paths for the $d{>}0$ (\hae{}-2/3) and $d{=}0$ (\hae{}-5) cases
(Section~\ref{sec:impl:lowstack-keygen}), and fits verification
within the \SI{8}{\kilo\byte} budget at every level.
Our pure~C implementation achieves
\SI{4.7}{\kilo\byte}--\SI{5.7}{\kilo\byte} key generation stack,
\SI{5.8}{\kilo\byte}--\SI{6.0}{\kilo\byte} signing stack, and
\SI{4.7}{\kilo\byte}--\SI{4.8}{\kilo\byte} verification stack
across all levels, with
\texttt{.data}\,$=$\,0 and \texttt{.bss}\,$=$\,0.

\paragraph{Novel stack optimization techniques.}
We introduce optimization techniques tailored to \hae{}'s algorithmic
structure that push the measured stack below the idealized streaming
baselines (Section~\ref{sec:theory}).
For signing (Section~\ref{sec:impl:lowstack-sign}), we introduce
(i)~\emph{Rejection-aware pass decomposition}, which confines large
buffers to non-overlapping lifetime phases;
(ii)~\emph{Component-level early rejection}, which short-circuits the
response computation when a partial norm exceeds the bound; and
(iii)~\emph{Reverse-order streaming entropy coding}, which eliminates
full hint and high-bits staging arrays.
The signing path additionally employs a streaming two-pass hyperball
sampler with a fully streaming Gaussian backend
(\texttt{.bss}\,$=$\,0), for which we provide a formal correctness
proof (Section~\ref{sec:sampler}).
For verification (Section~\ref{sec:impl:lowstack-verify}), we combine
row-streamed matrix multiplication with view-style decoding, replacing
a large vector accumulator with a single polynomial and a compact
union overlay.
For \hae{}-5, these yield \SI{4816}{\byte} key generation
($-7.6$\,\% vs.~\cite{cryptoeprint:2026/442}),
\SI{6136}{\byte} signing ($-24$\,\%), and \SI{4840}{\byte}
verification ($-22$\,\%), while key generation, signing, and
verification are 25\,\%, 27\,\%, and 16\,\% faster respectively.

\paragraph{Evaluation and practical applicability.}
We evaluate on the \texttt{pqm4} framework and additionally validate
portability under RIOT-OS on Cortex-M4 (nRF52840) and RISC-V
(ESP32-C6) targets.
On \SI{8}{\kilo\byte} devices, verification at all \hae{} security
levels fits within budget (signature buffer + stack), and is
$2.34$--$3.34\times$ faster than
ML-DSA m4fstack~\cite{DBLP:conf/africacrypt/BosRS22} at each
comparable security level.
On \SI{16}{\kilo\byte} devices, all \hae{} levels support full
signing (secret key + signature + stack), whereas ML-DSA-87 exceeds
this budget.

\paragraph{Constant-time implementation.}
Our implementation avoids branches and memory accesses that depend
on long-term secret coefficients, with the exception of the
rejection-sampling loop inherent in the Fiat-Shamir-with-Aborts
paradigm; verification processes only public inputs.
Countermeasures against power analysis or fault attacks are not
addressed in this work.
Our implementation will be available upon acceptance of the paper. 

%% file: background.tex
\section{Background}
\label{sec:background}

In this section, we briefly explain the mathematical and algorithmic
background of \hae{}, covering its algebraic setting, parameter sets,
and core operations.
Implementation-oriented algorithm specifications based on the latest
\hae{} specification~\cite{HAETAE_Final_spec} are provided in
Appendix~\ref{sec:app:keygen}--\ref{sec:app:verify}.

\subsection{HAETAE}
\label{sec:background:haetae}
\hae{} is a module-lattice signature scheme operating in the polynomial ring
\[
  R_q \;=\; \mathbb{Z}_q[X] / (X^N + 1),
\]
with $q$ a prime modulus.  Vectors of $\ell$
(resp.\ $k$) polynomials are called \texttt{polyvecl} (resp.\ \texttt{polyveck}) and
occupy $\ell \cdot N \cdot 4$ (resp.\ $k \cdot N \cdot 4$) bytes when stored with
32-bit coefficients.  Polynomial multiplication in $R_q$ is performed efficiently via
the Number Theoretic Transform~(NTT), which maps a polynomial to its evaluation at the
$2N$-th roots of unity modulo~$q$ and reduces multiplication to a pointwise product in
$\mathcal{O}(N \log N)$ time.

\paragraph{Parameter sets.}
\hae{} is specified at three security levels.
Table~\ref{tab:haetae-params} lists the parameters referenced in this
paper. All three sets share the polynomial degree $N{=}256$ and modulus
$q{=}64{,}513$.
The module dimensions $(k,\ell)$ determine the sizes of the public
matrix $\mathbf{A} \in R_q^{k \times \ell}$ and the secret/response
vectors, and are the primary driver of memory usage.
The truncation parameter $d$ governs the public-key format:
$d{=}1$ (HAETAE-2/3) stores a rounded representation, while
$d{=}0$ (HAETAE-5) stores the NTT-domain image directly
(Section~\ref{sec:background:keygen}).

\begin{table}[t]
  \centering
  \caption{Selected \hae{} parameters. $|\texttt{poly}| = N \times 4 = \SI{1024}{\byte}$.}
  \label{tab:haetae-params}
  \footnotesize
  \begin{tabular}{lcccc}
    \toprule
    Parameter & Description & \hae{}-2 & \hae{}-3 & \hae{}-5 \\
    \midrule
    $N$ & polynomial degree & 256 & 256 & 256 \\
    $q$ & modulus & 64\,513 & 64\,513 & 64\,513 \\
    $(k,\ell)$ & module dimensions & $(2,4)$ & $(3,6)$ & $(4,7)$ \\
    $\tau$ & challenge weight & 58 & 80 & 128 \\
    $d$ & PK truncation & 1 & 1 & 0 \\
    $\alpha$ & $\mathbf{z}_1$ compression & 256 & 256 & 256 \\
    $\alpha_h$ & $\mathbf{h}$ compression & 512 & 512 & 256 \\
    \midrule
    $|\textit{pk}|$ & public key (bytes) & 992 & 1\,472 & 2\,080 \\
    $|\textit{sk}|$ & secret key (bytes) & 1\,408 & 2\,112 & 2\,752 \\
    $|\textit{sig}|$ & signature (bytes) & 1\,474 & 2\,349 & 2\,948 \\
    \bottomrule
  \end{tabular}
\end{table}

Like ML-DSA~\cite{FIPS204}, \hae{} is constructed in the \emph{Fiat--Shamir with Aborts}
(FSwA) framework~\cite{Lyubashevsky09}. 
In FSwA, the signer commits to an
ephemeral randomness $\mathbf{y}$, receives (or derives) a binary challenge
$c$ of low Hamming weight, and computes a masked response $\mathbf{z} = \mathbf{y}
+ (-1)^b (c \star \mathbf{s})$.  An abort (rejection) step is performed to ensure
that the distribution of $\mathbf{z}$ does not leak information about the secret
$\mathbf{s}$.  The key distinction between ML-DSA and \hae{} lies in the
distribution from which the ephemeral $\mathbf{y}$ is drawn, as detailed in
Section~\ref{sec:background:dilithium}.

\hae{} makes extensive use of $\textsf{HighBits}$ and $\textsf{LowBits}$
decompositions; we abbreviate these as $\textsf{HB}$ and $\textsf{LB}$
throughout, with superscripts ($h$, $z_1$, $\mathit{pk}$) indicating the
decomposition type.
For HAETAE-2/3 ($d{=}1$), the verification key stores only the
high-order bits $\mathbf{b}_1 = \textsf{HB}^{pk}(\mathbf{b})$ together
with a public $\mathrm{seed}_{\mathbf{A}}$ for expanding the matrix.
The low-order bits
$\mathbf{b}_0 = \textsf{LB}^{pk}(\mathbf{b})$ are folded into
the secret key.
For HAETAE-5 ($d{=}0$), no rounding is applied;
instead $\widehat{\bpo} = \textsf{NTT}(-2\bpo)$ is stored directly
in the NTT domain.

The signature and key sizes of \hae{} are smaller than those of ML-DSA at
comparable security levels, owing to the hyperball-based sampling strategy
described below.  At HAETAE-2 (targeting NIST security level~2), the signature
size is 1\,474 bytes, making it competitive with other lattice-based signatures
while remaining efficient to verify.

\paragraph{Key Generation.}
\label{sec:background:keygen}

Key generation expands a seed into the public matrix
$\mathbf{A} \in R_q^{k \times \ell}$ and secret vectors
$(\mathbf{s}_1, \mathbf{s}_2)$, then computes
$\mathbf{b} = \mathbf{A}\mathbf{s}_1 + \mathbf{s}_2 \bmod q$.
Unlike ML-DSA, \hae{} applies an explicit singular-value norm check
$\mathcal{N}(\mathbf{s}_1, \mathbf{s}_2) \le \gamma^2 N$ that rejects
and resamples until the bound is satisfied; this rejection loop and its
FFT-based workspace affect the key-generation stack peak.
The implementation differs across security levels:
for HAETAE-2/3 ($d{>}0$), $\mathbf{b}$ is rounded and the public key
stores $\mathbf{b}_1 = \textsf{HB}^{pk}(\mathbf{b})$, while
for HAETAE-5 ($d{=}0$), $\widehat{\bpo} = \textsf{NTT}(-2\bpo)$ is
stored directly in the NTT domain.
The norm rejection also differs: HAETAE-2/3 interleave it with the
matrix--vector multiplication, while HAETAE-5 performs it before
expanding $\mathbf{A}$.
See Algorithm~\ref{alg:keygen_d_gt_0} and~\ref{alg:keygen_d_eq_0}
in Appendix~\ref{sec:app:keygen} for the detailed specifications.

\paragraph{Signing.}
\label{sec:background:sign}

Unlike ML-DSA, which samples $\mathbf{y}$ uniformly from a hypercube,
\hae{} draws $\mathbf{y} = (\mathbf{y}_1, \mathbf{y}_2)$ from a
\emph{discretized hyperball}: coefficients are sampled from a discrete
Gaussian $\mathcal{D}_\sigma$, the vector is rescaled so that
$\|(\mathbf{y}_1, \mathbf{y}_2)\|_2 \leq B_0\Lambda$, and a norm
rejection test is applied.
This yields shorter signatures but requires a more expensive and
memory-intensive sampler (Section~\ref{sec:sampler}).

In the signature operation, the signer derives $\mu = H(\pk, M)$ and repeats the steps of
sampling $\mathbf{y}$, computing
$\mathbf{w} = \mathbf{A}\lfloor\mathbf{y}_1\rceil \bmod 2q$, deriving
challenge $c$, and forming $\mathbf{z} = \mathbf{y} + (-1)^b c\star\mathbf{s}$.
If $\mathbf{z}$ passes the norm bounds, the hint
$\mathbf{h} = \textsf{HB}^h(\mathbf{w}') -
\textsf{HB}^h(\mathbf{w}' - 2\lfloor\mathbf{z}_2\rceil)$ is computed
and the signature
$\sigma = (\textsf{HB}^{z_1}(\lfloor\mathbf{z}_1\rceil),\,
\textsf{LB}^{z_1}(\lfloor\mathbf{z}_1\rceil),\,\mathbf{h},\,c)$
is encoded.
See Algorithm~\ref{alg:sign} in Appendix~\ref{sec:app:sign}.
A na\"ive implementation keeps $\mathbf{y}$, $\mathbf{w}$,
$\mathbf{z}$, and encoding buffers live simultaneously, totaling
$(\ell{+}2k)\cdot|\texttt{poly}|$ of ring objects
(e.g.\ over 12\,kB for \hae{}-5).
The hyperball sampler amplifies this: the scaling factor depends on the
norm of the full Gaussian sample, so a one-pass implementation must
store all $(L{+}K) \cdot N$ coefficients before producing any output.

\paragraph{Verification.}
\label{sec:background:verify}

Verification (Algorithm~\ref{alg:verify} in Appendix~\ref{sec:app:verify})
parses the signature, reconstructs
$\tilde{\mathbf{w}} = \mathbf{A}\tilde{\mathbf{z}}_1 \bmod 2q$,
recomputes the challenge from
$\tilde{\mathbf{h}} + \textsf{HB}^h(\tilde{\mathbf{w}}')$
and the message digest, and accepts if the recomputed challenge
matches and all norm bounds hold.
There is no rejection loop or secret-dependent sampler, but
materializing the matrix product, decoded signature, and hint buffers
simultaneously can still require substantial stack space.

\subsection{Applicable ML-DSA Memory Optimization}
\label{sec:background:dilithium}

Due to the similarity of the schemes, the memory-optimization techniques of~\cite{DBLP:conf/africacrypt/BosRS22} can be applied to
the operations common between \hae{} and ML-DSA.

Both schemes are module-lattice Fiat--Shamir signatures with the same high-level shape:
seed expansion to generate a public matrix $\mathbf{A}$; computation of a public
commitment involving $\mathbf{A}$ and short secret vectors; a sign-then-reject loop
in which an ephemeral sample is masked by a low-weight challenge; and verification by
recomputing the challenge from the response and the public key.  Both also rely on NTT-based
polynomial arithmetic and $\textsf{HB}/\textsf{LB}$ decompositions to
compress intermediate values.  This structural similarity means that the
\emph{streaming} and \emph{liveness-reduction} ideas from~\cite{DBLP:conf/africacrypt/BosRS22}
carry over to \hae{} with appropriate adaptation.

The primary algorithmic difference is the distribution from which the
ephemeral randomness $\mathbf{y} = (\mathbf{y}_1, \mathbf{y}_2)$ is drawn.
ML-DSA samples $\mathbf{y}$ \emph{uniformly from a hypercube}: each coefficient
is drawn independently and uniformly from $\{-\gamma_1 + 1, \dots, \gamma_1\}$.
\hae{} instead samples $\mathbf{y}$ from a \emph{discrete Gaussian conditioned on a
hyperball}: coefficients are drawn from a discrete Gaussian, the vector is rescaled
so that its Euclidean norm lies in $[\,0, B_0\Lambda\,]$, and a Euclidean norm
rejection test is applied.  This hyperball sampling yields shorter signatures and
keys because it more efficiently fills the acceptance region, but at the price of
a more expensive and memory-intensive sampler (see Section~\ref{sec:sampler}).

In ML-DSA, the public key is $(\rho, \mathbf{t}_1)$, where $\mathbf{t}_1$
is the high-order part of $\mathbf{A}\mathbf{s}_1 + \mathbf{s}_2$.  The matrix
seed $\rho$ is stored explicitly, and the secret key additionally holds $\mathbf{t}_0$
(the low-order residue) to enable efficient signing.
In \hae{}, the public key is $(\mathrm{seed}_{\mathbf{A}}, \mathbf{b}_1)$, where
$\mathbf{b}_1$ is the high-order part of a related lattice commitment $\mathbf{b}$.
The \hae{} secret key is more compact; it stores a single short vector $\mathbf{s}$
of augmented dimension $k + \ell + 1$ and a nonce $K$---because the lattice relation
allows the signer to recover the full signing basis on-the-fly.

\begin{figure}[t]
    \centering
    \includegraphics[width=0.9\linewidth]{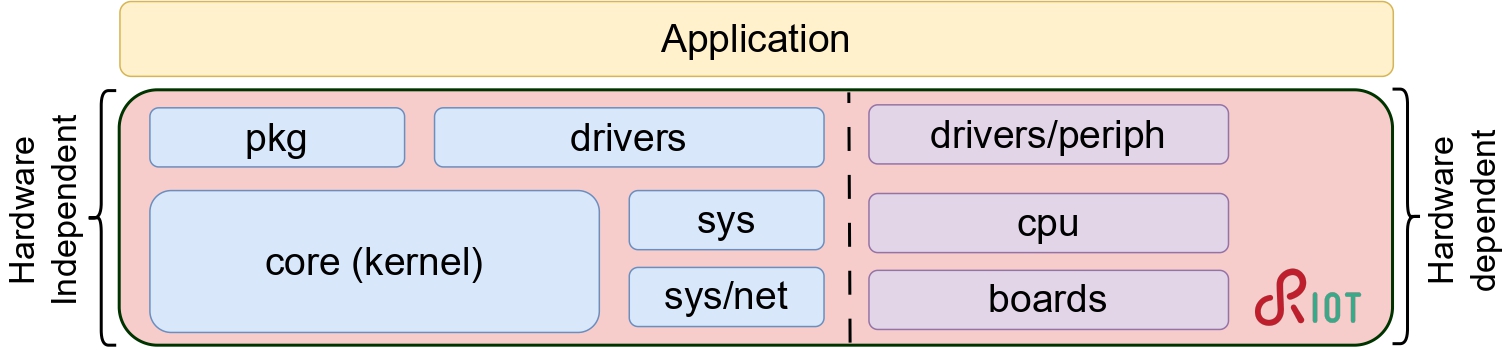}
    \vspace{-0.1cm}
    \caption{Overview of RIOT-OS modularization packages.}
    \label{fig:riot_modularization}
\end{figure}

\subsection{RIOT-OS}
RIOT~\cite{riot:home} is an open-source operating system targeting
resource-constrained IoT devices. Originally introduced
in~\cite{Riot_original}, it is written in C and designed with portability
in mind, which allows it to support a broad range of boards and
architectures. Figure~\ref{fig:riot_modularization} shows the modular
structure of RIOT. The most left part in the figure 
corresponds to hardware-independent components, whereas the right
part contains hardware-specific ones. This
design enables RIOT to reuse most of its code across platforms and to
favor portable C implementations instead of relying on architecture-specific
intrinsics or assembly code.

%% file: theoretical_eval.tex
\section{Idealized memory baselines for \hae{}{}}
\label{sec:theory}

Before describing the concrete implementation techniques in
Section~\ref{sec:techniques}, we analyze the idealized memory
footprint of each \hae{} routine by counting only the dominant ring
objects (polynomials and polynomial vectors) while ignoring constant-size
overheads such as hash states, scalar variables, and call frames.
This analysis clarifies the relationship between the number of
simultaneously live large ring objects and the achievable peak stack.

\paragraph{Dominant objects.}
All three \hae{} parameter sets use $N{=}256$; a single polynomial
with 32-bit coefficients occupies $|\texttt{poly}| = 1\,024$\,B.
The public matrix $\mathbf{A} \in R_q^{k \times \ell}$ has $k$ rows
and $\ell$ columns; a \texttt{polyveck} ($k$ polynomials) and
\texttt{polyvecl} ($\ell$ polynomials) scale linearly.
For \hae-5 $(k{=}4,\ell{=}7)$ these are 4\,096\,B and 7\,168\,B
respectively.

\paragraph{High-level algorithms.}
We summarize the \hae{} workflow in
Algorithm~\ref{alg:haetae-highlevel}; it is intentionally schematic since
our implementation keeps the mathematics unchanged while reorganizing
the order in which objects are materialized.
We use the $\textsf{HB}/\textsf{LB}$ abbreviations introduced in
Section~\ref{sec:background:haetae}.

\begin{algorithm}[t]
\caption{Schematic \hae{} (Key generation, Signing, Verification)}
\label{alg:haetae-highlevel}
\footnotesize
\begin{algorithmic}[1]
\Procedure{Key Generation}{$1^\lambda$}
  \State Sample seeds $\rho,\kappa$;
         sample $(\mathbf{s}_1,\mathbf{s}_2)$
  \State Reject if $\mathcal{N}(\mathbf{s}_1,\mathbf{s}_2) > \gamma^2 N$
         \Comment{FFT-based singular-value norm}
  \State Expand $\rho \mapsto \mathbf{A}$;
         compute $\mathbf{b} \gets \mathbf{A}\mathbf{s}_1 + \mathbf{s}_2 \bmod q$;
         output $(\pk,\sk)$
\EndProcedure
\Procedure{Signing}{$\sk, M$}
  \State Derive $\mu \gets H(\pk,M)$
  \Repeat
    \State Sample $(\mathbf{y}_1,\mathbf{y}_2)$; compute
           $\mathbf{w} \gets \mathbf{A}\lfloor\mathbf{y}_1\rceil \bmod 2q$
    \State $c \gets \textsf{SampleBinaryChallenge}_\tau(H(\textsf{HB}^h(\mathbf{w}'),
           \textsf{LSB}(\lfloor y_{1,1}\rceil), \mu))$
    \State $\mathbf{z} \gets \mathbf{y} + (-1)^b c\star\mathbf{s}$
  \Until{$\mathbf{z}$ passes rejection tests}
  \State $\mathbf{h} \gets \textsf{HB}^h(\mathbf{w}') -
         \textsf{HB}^h(\mathbf{w}' - 2\lfloor\mathbf{z}_2\rceil)$
         \Comment{post-acceptance}
  \State $\sigma \gets (\textsf{HB}^{z_1}(\lfloor\mathbf{z}_1\rceil),\,
         \textsf{LB}^{z_1}(\lfloor\mathbf{z}_1\rceil),\,\mathbf{h},\,c)$
\EndProcedure
\Procedure{Verification}{$\pk, M, \sigma$}
  \State Parse $\sigma$ into
         $(\textsf{HB}^{z_1},\textsf{LB}^{z_1},\mathbf{h},c)$;
         reconstruct $\tilde{\mathbf{z}}_1$
  \State $\tilde{\mathbf{w}} \gets \mathbf{A}\tilde{\mathbf{z}}_1 \bmod 2q$;
         recompute $c'$ from
         $(\tilde{\mathbf{h}}+\textsf{HB}^h(\tilde{\mathbf{w}}'),\mu)$
  \State Accept iff $c'{=}c$ and norm bounds hold
\EndProcedure
\end{algorithmic}
\end{algorithm}

\paragraph{Key generation.}
The reference implementation materializes $\mathbf{A}$, $\mathbf{s}_1$, and
$\mathbf{b}$ simultaneously
(${\approx}\,(k\ell{+}\ell{+}k)\cdot|\texttt{poly}|$).
Since both $\mathbf{A}$ and $\mathbf{s}_1$ are deterministically
expandable from seeds, the matrix--vector product can be streamed with
$2{\cdot}|\texttt{poly}|$ (one row accumulator plus one sampling scratch).
However, the singular-value norm check
$\mathcal{N}(\mathbf{s}_1,\mathbf{s}_2) \le \gamma^2 N$ requires an
FFT-based computation with a 2\,048\,B workspace, giving a total
baseline of $2{\cdot}|\texttt{poly}| + 2\,048 = 4\,096$\,B.

\paragraph{Signature generation.}
The reference implementation keeps $\mathbf{y}_1$ ($\ell$ polys),
$\mathbf{y}_2$ and $\mathbf{w}$ ($2k$ polys), and the
hint~$\mathbf{h}$ live simultaneously, totaling
${\approx}\,(\ell{+}2k)\cdot|\texttt{poly}|$
(e.g.\ $12\,288$\,B for \hae{}-5).
With standard streaming, $\mathbf{y}$ and $\mathbf{w}$ can be
regenerated row by row from seeds, but the full hint buffer
$\mathbf{h} \in R^k$, the high-bits array
$\textsf{HB}^{z_1}(\mathbf{z}_1) \in \mathbb{Z}^{\ell \times N}$,
and the challenge polynomial~$c$ must
still be materialized before the rANS encoder can process them.
In particular, $c$ remains live throughout the rejection test and
the packing phase, coexisting with the hint and streaming buffers.
This gives a streaming baseline of
$(k{+}3)\cdot|\texttt{poly}| + \ell{\cdot}N$ bytes
(e.g.\ $7\,168 + 1\,792 = 8\,960$\,B for \hae{}-5).

\paragraph{Verification.}
A column-streamed approach accumulates the matrix product
$\tilde{\mathbf{w}} = \widehat{\Apo}\circ\textsf{NTT}(\tilde{\mathbf{z}}_1)$
into a full \texttt{polyveck} buffer while sweeping
$\tilde{\mathbf{z}}_1$ once.
The challenge polynomial~$c$ must also remain live for the final
comparison, giving a baseline of
$(k{+}2)\cdot|\texttt{poly}|$ (e.g.\ $6\,144$\,B for \hae{}-5).

\paragraph{Summary and outlook.}
Table~\ref{tab:theory-vs-measured} summarizes the baseline streaming
footprints for \hae{}-5; these baselines count dominant ring objects
and staging buffers while excluding seeds, hash states, and other
implementation-specific overheads.
The same analysis applies to \hae{}-2/3 with adjusted $(k,\ell)$
values, though differences in the public-key domain ($d{>}0$ requires
an additional rounding step and $\mathbf{a}$-vector) and the placement
of the rejection loop affect the concrete figures.
The prior work of~\cite{cryptoeprint:2026/442} reports \hae{}-5
stack usage close to these baselines: \SI{5212}{\byte} (Key generation),
\SI{8092}{\byte} (Signing), and \SI{6220}{\byte} (Verification).
In the next section we introduce additional techniques, including
pass decomposition with \texttt{noinline} boundaries, reverse-order
streaming entropy coding, and row-streamed verification, that push the
stack footprint below these idealized baselines.
Section~\ref{sec:impl:results-pqm4} confirms that the resulting
implementation achieves lower stack usage while retaining competitive
cycle counts relative to~\cite{cryptoeprint:2026/442}.

\begin{table}[t]
  \centering
  \caption{Streaming baselines vs.\ measured stack on \texttt{pqm4}
  (\hae{}-5, bytes).
  Baselines count dominant ring objects and staging buffers;
  constant-size overheads (hash states, scalars, call frames) are
  excluded.
  Full performance results are in Table~\ref{tab:pqm4-eval}.}
  \label{tab:theory-vs-measured}
  \resizebox{0.8\linewidth}{!}{
  \begin{tabular}{lrrr}
    \toprule
    & Idealized baseline & \cite{cryptoeprint:2026/442} & Ours \\
    \midrule
    Key generation & $2 \times 1\,024 + 2\,048 = 4\,096$ & 5\,212 & \textbf{4\,816} \\
    Signing        & $(k{+}3) \times 1\,024 + \ell{\cdot}N = 8\,960$ & 8\,092 & \textbf{6\,136} \\
    Verification   & $(k{+}2) \times 1\,024 = 6\,144$ & 6\,220 & \textbf{4\,840} \\
    \bottomrule
  \end{tabular}
  }
\end{table}

%% file: techniques.tex
\section{Low-footprint memory techniques}
\label{sec:techniques}

In this section, we describe the main techniques we use to reduce the memory footprint 
of \hae{} key generation, signing, and verification. 
Our approach is guided by the principle of trading increased computation time for reduced peak stack footprint,
which is essential for predictable execution on memory-constrained microcontrollers. 
We also leverage the structure of \hae{}'s algorithms, 
such as the use of seed-derived objects and sparse challenges, 
to enable streaming and in-place computation strategies that minimize stack usage.

\subsection{Low-stack key generation via streaming}
\label{sec:impl:lowstack-keygen}

When key generation for digital signatures on constrained devices is required to be performed at runtime,
it can no longer be circumvented with a one-time provisioning step.
Practical applications include generating ephemeral keys for firmware integrity verification
and implementing key rotation to mitigate compromise.

The baseline key generation is specified in
Algorithms~\ref{alg:keygen_d_gt_0} and~\ref{alg:keygen_d_eq_0}
(Appendix~\ref{sec:app:keygen}).
We describe our stack optimizations in two parts.
First we present the streaming techniques applied to HAETAE-2/3 ($d > 0$), which involve rounding and an additional
$\apo$-vector not present in the $d{=}0$ case; this streaming of key generation
for $d > 0$ is a new contribution not covered in previous work.
Figure~\ref{fig:haetae-keygen-stream} illustrates the resulting memory layout.
We then describe our HAETAE-5 ($d = 0$) implementation, which
follows the same high-level approach as~\cite{cryptoeprint:2026/442}, but further compresses the stack through caller-level union analysis.

\paragraph{Streaming matrix--vector multiplication.}
Key generation computes
$\mathbf{b} = \mathbf{a} + \mathbf{A}\mathbf{s}_1 + \mathbf{s}_2 \pmod{q}$,
where $\mathbf{A}\in R_q^{k\times \ell}$ is derived from a seed and
$\mathbf{s}_1\in R_q^{\ell}$, $\mathbf{s}_2\in R_q^{k}$ are short secret vectors.
The reference implementation materializes $\mathbf{A}$ and $\widehat{\mathbf{s}}_1$
as full polynomial vectors, causing large stack pressure.
We reorganize the computation so that $\mathbf{s}_1$, $\mathbf{s}_2$, and $\mathbf{A}$
are all generated and consumed on-the-fly
(steps~\stepn{2}--\stepn{12} in Figure~\ref{fig:haetae-keygen-stream}):
for each column~$j$ we sample $\mathbf{s}_{1,j}$ into a single scratch polynomial
(\stepn{2}),
apply the NTT in place~(\stepn{4}), and for each row~$i$ generate $\mathbf{A}_{i,j}$ from the seed
and immediately accumulate the pointwise product into $\mathbf{b}_i$~(\stepn{5}).
The secret vectors are written to the secret key as they are produced~(\stepn{3}, \stepn{12}), and a
norm accumulator (\texttt{int32\_t sum[N]}) tracks the singular-value check
incrementally~(\stepn{13}); $\mathbf{s}_1$ norms are accumulated only at $i{=}0$ to
avoid redundant computation in the inner loop.
The entire key generation thus requires only two single-polynomial scratch
slots (\texttt{poly b, s}) and one norm accumulator on the stack.

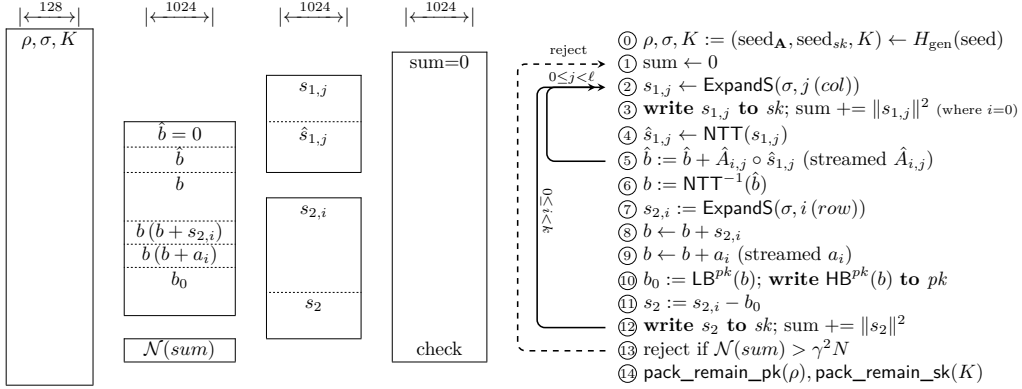
\begin{figure}[t!]
    \centering
    \resizebox{1.0\textwidth}{!}{
        \begin{tabular}{c p{0.05in} c p{0.05in} c p{0.05in} c p{0.05in} p{0.25in} p{0.05in} l}
            $\mid \xleftrightarrow{\;\; 128 \;\;}\mid$ & &
            $\mid \xleftrightarrow{\;\; 1024 \;\;}\mid$ & &
            $\mid \xleftrightarrow{\;\; 1024 \;\;}\mid$ & &
            $\mid \xleftrightarrow{\;\; 1024 \;\;}\mid$ & & & & \\
            \cline{1-1}
            \multicolumn{1}{|c|}{$\rho, \sigma, K$} & & & & & & & & & &
            \stepn{0} $\rho, \sigma, K := (\text{seed}_\Apo, \text{seed}_\sk, K) \gets H_{\text{gen}}(\text{seed})$ \\
            \cline{7-7}
            \multicolumn{1}{|c|}{} & & & & & &
            \multicolumn{1}{|c|}{$\text{sum}{=}0$} & &
            \tikzmark{rej2}{} & \tikzmark{rej3}{} &
            \stepn{1} $\text{sum} \gets 0$ \\
            \cline{5-5}
            \multicolumn{1}{|c|}{} & & & &
            \multicolumn{1}{|c|}{$s_{1,j}$} & &
            \multicolumn{1}{|c|}{} & &
            \tikzmark{lr2}{} & \tikzmark{lr3}{} &
            \stepn{2} $s_{1,j} \gets \textsf{ExpandS}(\sigma, j\,(col))$ \\
            \multicolumn{1}{|c|}{} & & & &
            \multicolumn{1}{|c|}{} & &
            \multicolumn{1}{|c|}{} & & & &
            \stepn{3} \textbf{write $s_{1,j}$ to $\sk$}; $\text{sum} \mathrel{+}= \|s_{1,j}\|^2$ {\scriptsize(where $i{=}0$)} \\
            \cline{3-3} \cdashline{5-5}[1pt/1pt]
            \multicolumn{1}{|c|}{} & & \multicolumn{1}{|c|}{$\hat{b}=0$} & &
            \multicolumn{1}{|c|}{$\hat{s}_{1,j}$} & &
            \multicolumn{1}{|c|}{} & & & &
            \stepn{4} $\hat{s}_{1,j} \gets \ntt(s_{1,j})$ \\
            \cdashline{3-3}[1pt/1pt]
            \multicolumn{1}{|c|}{} & &
            \multicolumn{1}{|c|}{$\hat{b}$} & &\multicolumn{1}{|c|}{} & &
            \multicolumn{1}{|c|}{} & &
            \tikzmark{cb1}{} & \tikzmark{cb0}{} &
            \stepn{5} $\hat{b} := \hat{b} + \hat{A}_{i,j} \circ \hat{s}_{1,j}$ (streamed $\hat{A}_{i,j}$)\\
            \cdashline{3-3}[1pt/1pt] \cline{5-5}
            \multicolumn{1}{|c|}{} & &
            \multicolumn{1}{|c|}{$b$} & & & &
            \multicolumn{1}{|c|}{} & & & &
            \stepn{6} $b := \invntt(\hat{b})$ \\
            \cline{5-5}
            \multicolumn{1}{|c|}{} & &
            \multicolumn{1}{|c|}{} & &
            \multicolumn{1}{|c|}{$s_{2,i}$} & &
            \multicolumn{1}{|c|}{} & & & &
            \stepn{7} $s_{2,i} := \textsf{ExpandS}(\sigma, i\,(row))$ \\
            \cdashline{3-3}[1pt/1pt]
            \multicolumn{1}{|c|}{} & &
            \multicolumn{1}{|c|}{$b\,(b + s_{2,i}$)} & &
            \multicolumn{1}{|c|}{} & &
            \multicolumn{1}{|c|}{} & & & &
            \stepn{8} $b \gets b + s_{2,i}$ \\
            \cdashline{3-3}[1pt/1pt]
            \multicolumn{1}{|c|}{} & &
            \multicolumn{1}{|c|}{$b\,(b + a_i$)} & &
            \multicolumn{1}{|c|}{} & &
            \multicolumn{1}{|c|}{} & & & &
            \stepn{9} $b \gets b + a_{i}$ (streamed $a_i$) \\
            \cdashline{3-3}[1pt/1pt]
            \multicolumn{1}{|c|}{} & &
            \multicolumn{1}{|c|}{$b_0$} & &
            \multicolumn{1}{|c|}{} & &
            \multicolumn{1}{|c|}{} & & & &
            \stepn{10} $b_0 := \textsf{LB}^{pk}(b)$; \textbf{write} $\textsf{HB}^{pk}(b)$  \textbf{to $\pk$} \\
            \cdashline{5-5}[1pt/1pt]
            \multicolumn{1}{|c|}{} & &
            \multicolumn{1}{|c|}{} & &
            \multicolumn{1}{|c|}{$s_2$} & &
            \multicolumn{1}{|c|}{} & & & &
            \stepn{11} $s_2 := s_{2,i} - b_0$ \\
            \cline{3-3}
            \multicolumn{1}{|c|}{} & & & &
            \multicolumn{1}{|c|}{} & &
            \multicolumn{1}{|c|}{} & &
            \tikzmark{rb1}{} & \tikzmark{rb0}{} &
            \stepn{12} \textbf{write $s_2$ to $\sk$}; $\text{sum} \mathrel{+}= \|s_2\|^2$ \\
            \cline{3-3} \cline{5-5}
            \multicolumn{1}{|c|}{} & &
			\multicolumn{1}{|c|}{$\mathcal{N}(sum)$} & & & &
            \multicolumn{1}{|c|}{check} & &
            \tikzmark{rej1}{} & \tikzmark{rej0}{} &
            \stepn{13} reject if $\mathcal{N}(sum) > \gamma^2 N$ \\
            \cline{3-3} \cline{7-7}
            \multicolumn{1}{|c|}{} & & & & & & & & & &
            \stepn{14} $\textsf{pack\_remain\_pk}(\rho), \textsf{pack\_remain\_sk}(K)$ \\
            \cline{1-1}
        \end{tabular}
        \begin{tikzpicture}[remember picture, overlay, >=stealth, shift={(0,0)}, thick]            
            \draw[rounded corners,->] ([yshift=\tikzoffset, xshift=5pt] cb0.east) -- ([yshift=\tikzoffset, xshift=5pt] cb1.east) -- ([yshift=\tikzoffset, xshift=5pt] lr2.east) -- ([yshift=\tikzoffset, xshift=5pt] lr3.east) node [midway, above, sloped, xshift=-4pt, yshift=-2pt] () {\ ${\scriptstyle 0\le j < \ell}$};
                        
            \draw[rounded corners,->] ([yshift=\tikzoffset, xshift=5pt] rb0.east) -- ([yshift=\tikzoffset] rb1.east) -- ([yshift=\tikzoffset] lr2.east) node [midway, right, rotate=-90, anchor=south, yshift=-2pt] () {\ ${\scriptstyle 0\le i < k}$}-- ([yshift=\tikzoffset] lr3.east);
            
            \draw[rounded corners, dashed, ->] ([yshift=\tikzoffset, xshift=5pt] rej0.east) -- ([yshift=\tikzoffset, xshift=-10pt] rej1.east) -- ([yshift=\tikzoffset, xshift=-10pt] rej2.east) -- ([yshift=\tikzoffset, xshift=5pt] rej3.east) node [midway, above, sloped, xshift=2pt] () {\ ${\scriptstyle \text{reject}}$};
        \end{tikzpicture}}
    \caption{
        Memory allocation of the proposed memory-optimized \hae{}-2,3 key generation with streaming.
        This approach streams the matrix $\hat{A}$ on-the-fly and reuses memory slots, keeping only two single polynomials (\texttt{poly b, s}) and one norm accumulator (\texttt{int32\_t sum[N]}) on the stack.
        The dashed arrow indicates the rejection loop unique to \hae{}-2,3.
    }
    \label{fig:haetae-keygen-stream}
\end{figure}

\paragraph{Additional stack reductions.}
SHAKE-based sampling routines (\texttt{poly\_uniform}, and 
\texttt{poly\_uniform\_eta}) are rewritten to squeeze one block at a time, replacing multi-block stack buffers
with a single-block buffer without altering the output distribution.
Together with the streaming matrix multiplication and the fused frozen-$A$
variant, these techniques remove full-vector temporaries, stack-local expansion
buffers, and multi-block squeeze buffers, yielding a substantially smaller and
more predictable stack profile.
Note that streaming and re-computation alter data-access patterns; a thorough
side-channel analysis is left for future work.

\paragraph{HAETAE-5 ($d = 0$) key generation.}
For HAETAE-5, the key generation follows the same high-level streaming
approach as prior work~\cite{cryptoeprint:2026/442}: on-the-fly matrix generation,
row-by-row accumulation of $\widehat{\bpo} = \textsf{NTT}(-2\bpo)$, and
direct storage in the NTT domain.
The two-pass structure is illustrated
in Figure~\ref{fig:haetae5-keygen-stream} (Appendix~\ref{sec:app:keygen}):
Pass~1 (norm check, steps~\stepn{1}--\stepn{6} in Figure~\ref{fig:haetae5-keygen-stream})
and Pass~2 (matrix--vector product, steps~\stepn{7}--\stepn{17}).
Our implementation additionally introduces two caller-level \texttt{union}s that
tightly pack all large temporaries into a fixed-size caller frame,
eliminating the need for deep callee stack frames.
The first union shares the 2\,048-byte FFT scratch used for the singular-value
norm in Pass~1~(\stepn{3}, \stepn{5}) with the 1\,024-byte sampling polynomial (\texttt{poly~s})
used in both passes~(\stepn{2}, \stepn{7}); the FFT is invoked directly in the caller rather
than delegated to a callee, so the scratch never appears as a separate
stack frame.
The second union shares the norm accumulator (\texttt{int32\_t sum[N]},
\SI{1024}{\byte}) in Pass~1~(\stepn{1}) with the row accumulator (\texttt{poly~b}, \SI{1024}{\byte})
in Pass~2~(\stepn{10}), exploiting the strict lifetime disjointness across the two passes.
Together, the two unions compress the dominant temporaries to
$2\,048 + 1\,024 = 3\,072$\,bytes in the caller frame, with no additional
callee overhead for the norm computation.
By analyzing the stack frames of each callee and identifying
lifetime-disjoint buffers across the two passes, this union-based layout
reduces the key-generation stack by approximately \SI{400}{\byte} compared
to~\cite{cryptoeprint:2026/442} (\SI{4784}{\byte} vs.\ \SI{5212}{\byte}
on the same \texttt{pqm4}framework).

\begin{figure}[t]
    \centering
    \resizebox{1.0\textwidth}{!}{
        \begin{tabular}{c p{0.05in} c p{0.05in} c p{0.05in} c p{0.05in} p{0.25in} p{0.05in} l}
            $\mid \xleftrightarrow{\;\; 160 \;\;}\mid$ & &
            $\mid \xleftrightarrow{\;\; 1024 \;\;}\mid$ & &
            $\mid \xleftrightarrow{\;\; 1024 \;\;}\mid$ & &
            $\mid \xleftrightarrow{\;\; 1024 \;\;}\mid$ & & & & \\[2pt]
            \cline{1-1}
            \multicolumn{1}{|c|}{seeds, $\mu$} & & & & & & & & & &
            \stepn{0} derive $\mu$, $\mathrm{seed}_{ybb}$; $\kappa \gets 0$ \\[3pt]
            \multicolumn{1}{|c|}{} & & & & & & & &
            \tikzmark{sRejDst2}{} & \tikzmark{sRejDst3}{} &
            {\scriptsize\textbf{--- Pass~A: row-streaming $\mathbf{w}$ and challenge ---}} \\[2pt]
            \cline{3-3} \cline{5-5}
            \multicolumn{1}{|c|}{} & &
            \multicolumn{1}{|c|}{$\hat{w}_i{=}0$} & &
            \multicolumn{1}{|c|}{$\mathrm{ytmp}$} & & & &
            \tikzmark{sj2}{} & \tikzmark{sj3}{} &
            \stepn{1} $\mathrm{ytmp} \gets \textsf{ExpandYbb}(\mathrm{seed}_{ybb}, \kappa)_j$ \\[2pt]
            \cdashline{5-5}[1pt/1pt]
            \multicolumn{1}{|c|}{} & &
            \multicolumn{1}{|c|}{} & &
            \multicolumn{1}{|c|}{$\widehat{\mathrm{ytmp}}$} & & & &
            \tikzmark{sj1}{} & \tikzmark{sj0}{} &
            \stepn{2} $\hat{w}_i \mathrel{+}= \widehat{\Apo}_{i,j} \circ \textsf{NTT}(\mathrm{ytmp})$ \\[1pt]
            \multicolumn{1}{|c|}{} & &
            \multicolumn{1}{|c|}{} & &
            \multicolumn{1}{|c|}{} & & & & & &
            {\scriptsize $\triangleright$ streamed $\widehat{\Apo}$} \\[2pt]
            \cdashline{3-3}[1pt/1pt] \cline{5-5}
            \multicolumn{1}{|c|}{} & &
            \multicolumn{1}{|c|}{$\mathbf{w}_i$} & & & & & & & &
            \stepn{3} $\mathbf{w}_i \gets \textsf{NTT}^{-1}(\hat{w}_i) + 2 \cdot \lfloor y_{2,i} \rceil \bmod q$ \\[2pt]
            \cdashline{3-3}[1pt/1pt]
            \multicolumn{1}{|c|}{} & &
            \multicolumn{1}{|c|}{$\mathbf{w}'_i$} & & & & & & & &
            \stepn{4} $\mathbf{w}'_i \gets \textsf{fromCRT}(\mathbf{w}_i, \lfloor y_{1,1} \rceil)$ \\[2pt]
            \multicolumn{1}{|c|}{} & &
            \multicolumn{1}{|c|}{} & & & & & & & &
            \stepn{5} $\mathbf{w}'_{1,i} \gets \textsf{HB}^h(\mathbf{w}'_i)$ \\[2pt]
            \multicolumn{1}{|c|}{} & &
            \multicolumn{1}{|c|}{} & & & & & &
            \tikzmark{si1}{} & \tikzmark{si0}{} &
            \stepn{6} absorb $\mathbf{w}'_{1,i}$ into $\mathcal{H}$ \\[1pt]
            \multicolumn{1}{|c|}{} & &
            \multicolumn{1}{|c|}{} & & & & & & & &
            {\scriptsize $\triangleright$ incremental transcript hashing} \\[2pt]
            \cline{3-3}
            \multicolumn{1}{|c|}{} & & & & & & & & & &
            \stepn{7} $\rho \gets H(\mathbf{w}'_1, \textsf{LSB}(\lfloor y_{1,1} \rceil), \mu)$ \\[2pt]
            \cline{3-3}
            \multicolumn{1}{|c|}{} & &
            \multicolumn{1}{|c|}{$c$} & & & & & & & &
            \stepn{8} $c \gets \textsf{SampleBinaryChallenge}_\tau(\rho)$ \\[3pt]
            \multicolumn{1}{|c|}{} & &
            \multicolumn{1}{|c|}{} & & & & & & & &
            {\scriptsize\textbf{--- Pass~B: streaming rejection check ---}} \\[2pt]
            \cline{5-5} \cline{7-7}
            \multicolumn{1}{|c|}{} & &
            \multicolumn{1}{|c|}{} & &
            \multicolumn{1}{|c|}{$\mathbf{y}_i$} & &
            \multicolumn{1}{|c|}{$\widehat{\spo}_i$} & &
            \tikzmark{sBi2}{} & \tikzmark{sBi3}{} &
            \stepn{9} $\mathbf{z}_i \gets \mathbf{y}_i + (-1)^b \textsf{NTT}^{-1}(\widehat{c} \circ \widehat{\spo}_i)$ \\[1pt]
            \multicolumn{1}{|c|}{} & &
            \multicolumn{1}{|c|}{} & &
            \multicolumn{1}{|c|}{} & &
            \multicolumn{1}{|c|}{} & &
            \tikzmark{sEE1}{} & \tikzmark{sEE0}{} &
            {\scriptsize $\triangleright$ fused; early reject if $\|z_{1,0}\|^2 > B_1^2$} \\[2pt]
            \cdashline{5-5}[1pt/1pt] \cline{7-7}
            \multicolumn{1}{|c|}{} & &
            \multicolumn{1}{|c|}{} & &
            \multicolumn{1}{|c|}{$\mathbf{z}_i$} & & & &
            \tikzmark{sBi1}{} & \tikzmark{sBi0}{} &
            \stepn{10} $\|\mathbf{z}\|^2 \mathrel{+}= \|\mathbf{z}_i\|^2$; discard $\mathbf{z}_i$ \\[2pt]
            \cline{5-5}
            \multicolumn{1}{|c|}{} & &
            \multicolumn{1}{|c|}{} & & & & & &
            \tikzmark{sRej1}{} & \tikzmark{sRej0}{} &
            \stepn{11} (early) reject if $\|(\mathbf{z}_1,\mathbf{z}_2)\|_2 \geq B'$; repeat Pass~A,\,B \\[4pt]
            \multicolumn{1}{|c|}{} & &
            \multicolumn{1}{|c|}{} & & & & & & & &
            {\scriptsize\textbf{--- Pass~C: post-acceptance (\textsf{noinline}) ---}} \\[2pt]
            \cline{5-5} \cline{7-7}
            \multicolumn{1}{|c|}{} & &
            \multicolumn{1}{|c|}{} & &
            \multicolumn{1}{|c|}{$\mathbf{w}_i$} & &
            \multicolumn{1}{|c|}{$\mathbf{z}_{2,i}$} & &
            \tikzmark{sC1i2}{} & \tikzmark{sC1i3}{} &
            \stepn{12} $C_1$: $\mathbf{h}_i \gets \mathbf{w}'_{1,i} - \textsf{HB}^h(\mathbf{w}'_i - 2\lfloor\mathbf{z}_{2,i}\rceil)$ \\[1pt]
            \multicolumn{1}{|c|}{} & &
            \multicolumn{1}{|c|}{} & &
            \multicolumn{1}{|c|}{} & &
            \multicolumn{1}{|c|}{} & &
            \tikzmark{sC1i1}{} & \tikzmark{sC1i0}{} &
            {\scriptsize $\triangleright$ $h_{i,j}$ streamed to rANS encoder (no $\mathbf{h}$ buffer); $i{=}K{-}1{\downarrow}0$} \\[2pt]
            \cline{5-5} \cline{7-7}
            \cline{5-5} \cline{7-7}
            \multicolumn{1}{|c|}{} & &
            \multicolumn{1}{|c|}{} & &
            \multicolumn{1}{|c|}{$\lfloor\mathbf{z}_{1,i}\rceil$} & &
            \multicolumn{1}{|c|}{$\widehat{\spo}_{1,i}$} & &
            \tikzmark{sC2i2}{} & \tikzmark{sC2i3}{} &
            \stepn{13} $C_2$: $\sigma \gets \textsf{PackSig}(\textsf{HB}^{z_1}(\lfloor\mathbf{z}_{1,i}\rceil), \textsf{LB}^{z_1}(\lfloor\mathbf{z}_{1,i}\rceil))$ \\[1pt]
            \multicolumn{1}{|c|}{} & &
            \multicolumn{1}{|c|}{} & &
            \multicolumn{1}{|c|}{} & &
            \multicolumn{1}{|c|}{} & &
            \tikzmark{sC2i1}{} & \tikzmark{sC2i0}{} &
            {\scriptsize $\triangleright$ $\textsf{HB}(z_{1,i,j})$ streamed to rANS encoder (no $\textsf{HB}$ buffer); $i{=}L{-}1{\downarrow}0$} \\[2pt]
            \cline{3-3} \cline{5-5} \cline{7-7}
            \cline{1-1}
            & & & & & & & & & &
            \stepn{14} $\textsf{HB}^{z_1}$ copy-out at end of $C_2$; $\mathbf{h}$ appended at finalization \\
        \end{tabular}
        \begin{tikzpicture}[remember picture, overlay, >=stealth, shift={(0,0)}, thick]
            \draw[rounded corners,->] ([yshift=\tikzoffset, xshift=5pt] sj0.east) -- ([yshift=\tikzoffset, xshift=5pt] sj1.east) -- ([yshift=\tikzoffset, xshift=5pt] sj2.east) -- ([yshift=\tikzoffset, xshift=5pt] sj3.east) node [midway, above, sloped, xshift=-4pt, yshift=-2pt] () {\ ${\scriptstyle 0\le j < L}$};
            \draw[rounded corners,->] ([yshift=\tikzoffset, xshift=5pt] si0.east) -- ([yshift=\tikzoffset] si1.east) -- ([yshift=\tikzoffset] sj2.east) node [midway, right, rotate=-90, anchor=south, yshift=-2pt] () {\ ${\scriptstyle 0\le i < K}$} -- ([yshift=\tikzoffset] sj3.east);
            \draw[rounded corners, dashed, ->] ([yshift=\tikzoffset, xshift=5pt] sRej0.east) -- ([yshift=\tikzoffset, xshift=-15pt] sRej1.east) -- ([yshift=\tikzoffset, xshift=-15pt] sRejDst2.east) -- ([yshift=\tikzoffset, xshift=5pt] sRejDst3.east) node [midway, above, sloped, xshift=2pt] () {\ ${\scriptstyle \text{reject}}$};
            \draw[rounded corners,->] ([yshift=\tikzoffset, xshift=5pt] sBi0.east) -- ([yshift=\tikzoffset, xshift=5pt] sBi1.east) -- ([yshift=\tikzoffset, xshift=5pt] sBi2.east) -- ([yshift=\tikzoffset, xshift=5pt] sBi3.east) node [midway, above, sloped, xshift=-4pt, yshift=-2pt] () {\ ${\scriptstyle 0\le i < L{+}K}$};
            \draw[rounded corners, dashed, ->] ([yshift=\tikzoffset, xshift=5pt] sEE0.east) -- ([yshift=\tikzoffset, xshift=-15pt] sEE1.east) -- ([yshift=\tikzoffset, xshift=-15pt] sRejDst2.east) -- ([yshift=\tikzoffset, xshift=5pt] sRejDst3.east);
            \draw[rounded corners,->] ([yshift=\tikzoffset, xshift=5pt] sC1i0.east) -- ([yshift=\tikzoffset, xshift=5pt] sC1i1.east) -- ([yshift=\tikzoffset, xshift=5pt] sC1i2.east) -- ([yshift=\tikzoffset, xshift=5pt] sC1i3.east) node [midway, above, sloped, xshift=-4pt, yshift=-2pt] () {\ ${\scriptstyle K{-}1 \ge i \ge 0}$};
            \draw[rounded corners,->] ([yshift=\tikzoffset, xshift=5pt] sC2i0.east) -- ([yshift=\tikzoffset, xshift=5pt] sC2i1.east) -- ([yshift=\tikzoffset, xshift=5pt] sC2i2.east) -- ([yshift=\tikzoffset, xshift=5pt] sC2i3.east) node [midway, above, sloped, xshift=-4pt, yshift=-2pt] () {\ ${\scriptstyle L{-}1 \ge i \ge 0}$};
        \end{tikzpicture}}
    \caption{
        Memory allocation of the proposed pass-decomposed \hae{} signing.
        The driver holds persistent state (seeds, $\mu$; 160\,B) and
        a single 1\,024\,B polynomial slot that serves as the row
        accumulator $\hat{w}_i$ for the matrix--vector product in Pass~A
        and holds the challenge~$c$ from Pass~B onward.
        Two additional poly-sized scratch slots (1\,024\,B each) are
        reused across passes.
        The dashed arrow marks the rejection loop (Pass~A\,+\,B only);
        Pass~C runs once after acceptance.
        $C_1$ and $C_2$ are invoked via \texttt{noinline},
        so their frames never coexist.
        Reverse-order loops ($i{\downarrow}, j{\downarrow}$) in $C_1$/$C_2$
        stream coefficients directly to the rANS encoder, eliminating full
        $\mathbf{h}$ and $\textsf{HB}(\mathbf{z}_1)$ staging buffers.
    }
    \label{fig:sign-stream}
\end{figure}
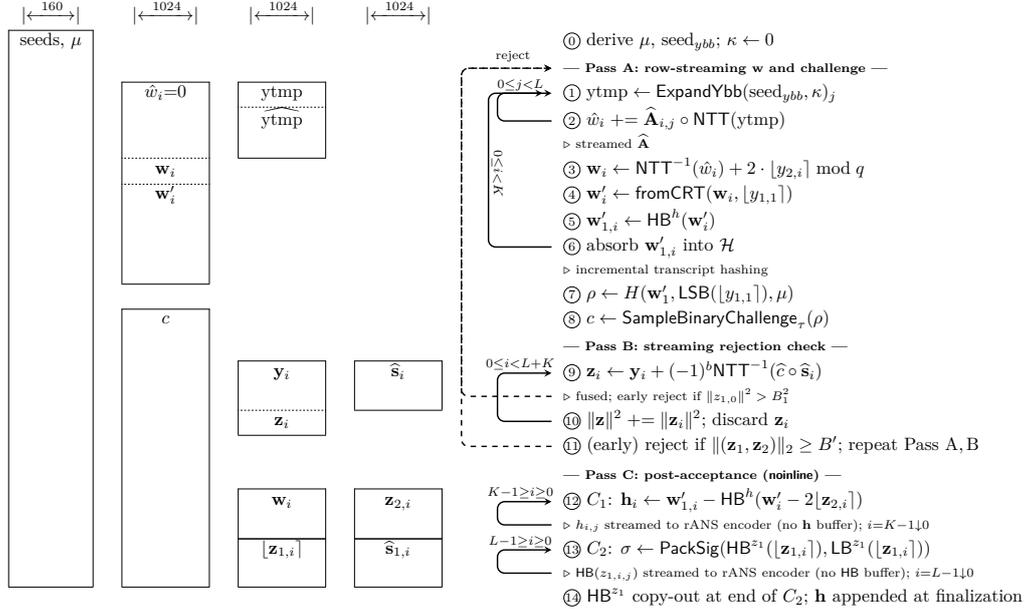

\subsection{Low-stack signing via pass decomposition and reverse-order streaming}
\label{sec:impl:lowstack-sign}

Recall \hae{} signing (Algorithm~\ref{alg:sign} in Appendix~\ref{sec:app:sign})
mainly consists of a rejection-sampling loop: each iteration draws ephemeral vectors
$(\mathbf{y}_1,\mathbf{y}_2)$ from a hyperball distribution, computes the
matrix--vector product
$\mathbf{w} \gets \textsf{NTT}^{-1}(\widehat{\Apo}\circ
\textsf{NTT}(\lfloor\mathbf{y}_1\rceil)) + 2\lfloor\mathbf{y}_2\rceil
\bmod q$, derives a challenge~$c$, forms the response
$\mathbf{z} = \mathbf{y} + (-1)^b c\star\mathbf{s}$, and repeats until the
norm bounds are satisfied.
The reference implementation~\cite{NISTPQC-ADD-R1:HAETAE23} keeps
$\mathbf{w}$, $\mathbf{z}$, and the hint~$\mathbf{h}$ live simultaneously,
leading to high peak stack usage.

Our design follows the~\cite{DBLP:conf/africacrypt/BosRS22} principle of
trading recomputation for reduced peak memory, and adopts the seed-based
on-the-fly matrix streaming used by~\cite{cryptoeprint:2026/442}: each
entry of $\widehat{\Apo}$ is regenerated from the public seed rather than
stored in RAM.
Beyond these shared foundations, we introduce several techniques not
described by~\cite{cryptoeprint:2026/442}:
(i)~a \emph{Rejection-aware pass decomposition} that confines encoding and
hint computation to a post-acceptance path;
(ii)~\emph{Component-level early rejection} that short-circuits
the response computation in Pass~B when a partial norm already exceeds the
bound;
and (iii)~\emph{Reverse-order streaming entropy coding} that eliminates full
hint and high-bit staging buffers by aligning computation order with the
rANS encoder's backward symbol consumption.
Additionally, our signing path achieves zero mutable BSS
(\texttt{.bss}\,$=$\,0\,B) through a streaming Gaussian backend
(Section~\ref{sec:sampler}).
Moreover, while~\cite{cryptoeprint:2026/442} targets \hae{}-5 only, our
design supports all three security levels including \hae{}-2/3 ($d{>}0$).
Figure~\ref{fig:sign-stream} illustrates the resulting memory layout.

\paragraph{Rejection-aware pass decomposition.}
We decompose the signing loop into three passes whose large buffers have
strictly non-overlapping lifetimes:
\begin{itemize}
\item \textbf{Pass~A} (steps~\stepn{1}--\stepn{8} in Figure~\ref{fig:sign-stream})
      samples $(\mathbf{y}_1,\mathbf{y}_2)$ via the
      two-pass hyperball sampler (Section~\ref{sec:sampler}), computes
      $\mathbf{w}$ row by row using on-the-fly matrix streaming, derives
      $\mathbf{w}'_1 = \textsf{HB}^h(\textsf{fromCRT}(\mathbf{w},
      \lfloor y_{1,1}\rceil))$, and obtains the challenge
      $c \gets \textsf{SampleBinaryChallenge}_\tau\!\bigl(H(\mathbf{w}'_1,
      \textsf{LSB}(\lfloor y_{1,1}\rceil), \mu)\bigr)$.
      The single-polynomial row accumulator for $\mathbf{w}$ reuses the
      output location of~$c$ (overwritten only at the end of the pass),
      eliminating a dedicated scratch buffer.
      For $d{>}0$ (\hae{}-2/3), the first column of $\widehat{\Apo}$ is
      derived on-the-fly from the packed public key via in-place CRT
      reconstruction, requiring no additional polynomial buffer.
\item \textbf{Pass~B} (\stepn{9}--\stepn{11}) regenerates $(\mathbf{y}_1,\mathbf{y}_2)$ from the
      stored seed and accepted nonce, computes
      $\mathbf{z} = \mathbf{y} + (-1)^b c\star\mathbf{s}$ one component at a
      time via fused accumulation (described below), and evaluates the
      rejection tests.
      No matrix product is needed, since the $\ell_2$-norm and
      $\ell_\infty$ tests involve only $\mathbf{z}$, $\mathbf{y}$,
      and~$b'$.
      Each component is discarded after its norm contribution is accumulated.
\item \textbf{Pass~C} (\stepn{12}--\stepn{14}) runs only once, after acceptance.
      It recomputes $\mathbf{z}$ and $\mathbf{w}$ to derive the hint
      $\mathbf{h} \gets \mathbf{w}'_1 -
       \textsf{HB}^h(\mathbf{w}' - 2\lfloor\mathbf{z}_2\rceil)
       \bmod^{+} \frac{2(q{-}1)}{\alpha_h}$
      and packs the signature.
      The variable-length payloads $\mathbf{h}$ and
      $\textsf{HB}^{z_1}(\mathbf{z}_1)$ are encoded via reverse-order streaming
      rANS (described below).
      Pass~C is further split into two sub-phases ($C_1$, hint encoding, and $C_2$, $\mathbf{z}_1$ packing) with non-overlapping buffer lifetimes.
\end{itemize}
Because the rejection loop comprises only Pass~A and Pass~B, all encoding
and hint computation are deferred to the post-acceptance path.
The peak signing stack is therefore
\begin{equation}\label{eq:sign-peak}
  S_{\mathrm{sign}} \;=\;
  S_{\mathrm{driver}} + \max\!\bigl(S_A,\; S_B,\; S_{C_1},\; S_{C_2}\bigr),
\end{equation}
where $S_{\mathrm{driver}}$ is the persistent driver state (seeds, nonce
counter, hash prefix) and the four terms are the transient allocations of
each pass.
Our decomposition guarantees that the peak follows
Equation~\eqref{eq:sign-peak} through explicit \texttt{noinline}
boundaries between passes, confining encoding and hint computation to the
single post-acceptance execution.

In contrast, the idealized streaming baseline from
Section~\ref{sec:theory} requires the full hint and high-bits arrays
to be materialized before encoding:
\[
  S_{\mathrm{sign}}^{\mathrm{conv}} \;=\;
  (k{+}2)\cdot|\texttt{poly}| + \ell \cdot N.
\]
Our pass-aware design replaces this with
Equation~\eqref{eq:sign-peak}, where for \hae{}-5:
\[
  S_{C_1} \approx 3\cdot|\texttt{poly}| + B_h,
  \qquad
  S_{C_2} \approx 2\cdot|\texttt{poly}| + B_{hb},
\]
with $B_h$ and $B_{hb}$ denoting the compact streaming encoder
buffers (bounded by the base entropy of $\mathbf{h}$ and
$\textsf{HB}^{z_1}(\mathbf{z}_1)$ respectively), which are
substantially smaller than the full staging arrays they replace.
This means we are even below the idealized streaming baseline.
Since $S_{C_1}$ and $S_{C_2}$ dominate the peak while Pass~A requires
only a single polynomial accumulator, the pass-aware structure leaves
unused stack budget in Pass~A.
We exploit this by batching two matrix rows per iteration of the
inner product, halving the number of ephemeral-vector regenerations
from $k\ell$ to $\lceil k/2 \rceil \cdot \ell$
at the cost of one additional polynomial accumulator
($+1\,024$\,B in~$S_A$), without increasing the overall
peak~$S_{\mathrm{sign}}$.

\paragraph{Component-level early rejection.}
The reference implementation computes the full response vector
$\mathbf{z} = (\mathbf{z}_1, \mathbf{z}_2)$ before evaluating any
rejection condition.
We introduce component-level early exits that reduce the average cost
of rejected iterations in Pass~B
(steps~\stepn{9}--\stepn{10} in Figure~\ref{fig:sign-stream}).
The first component $z_{1,0} = y_{1,0} + (-1)^b \cdot c$ involves
no secret-key material (the challenge~$c$ is derived from public
data), so a data-dependent branch on its norm is constant-time safe.
Pass~B therefore computes $z_{1,0}$ before any NTT-based
multiplication~(\stepn{9}) and rejects immediately if
$\|z_{1,0}\|^2 > B_1^2$ (dashed arrow between
\stepn{9} and~\stepn{10}), skipping the remaining
$\ell {+} k {-} 1$ fused multiply-accumulate operations.
During the subsequent per-component norm accumulation~(\stepn{10})
$\|\mathbf{z}\|^2 = \sum_i \|z_i\|^2$, we check the running partial
sum after each polynomial and skip the remaining components once the
bound is exceeded.
Additionally, the hyperball scaling factors, which are deterministic
functions of the seed and nonce counter, are computed once in the
driver frame and passed by pointer to Passes~B, $C_1$, and~$C_2$,
eliminating three redundant recomputations per accepted iteration.

\paragraph{Reverse-order streaming entropy coding.}
The signature includes two variable-length rANS-encoded payloads: the hint
$\mathbf{h} \in \mathbb{Z}^{k \times N}$ and the high-bit decomposition
$\textsf{HB}(\mathbf{z}_1) \in \mathbb{Z}^{\ell \times N}$.
Conventional implementations materialize these as full arrays before
encoding; for \hae{}-5, a \texttt{polyveck} hint buffer
($k {\times}\, 1\,024$\,B~$= 4\,096$\,B) and an \texttt{int8\_t} array
($\ell {\times}\, N = 1\,792$\,B) for $\textsf{HB}(\mathbf{z}_1)$.

rANS is a \emph{backward} entropy coder: for a flat array
$X[0], \ldots, X[m{-}1]$, the encoder consumes symbols in the order
$X[m{-}1], X[m{-}2], \ldots, X[0]$.
The conventional approach first materializes the entire array, then
iterates backward:
\[
  \textsf{RansEncPut}(X[m{-}1]),\;
  \textsf{RansEncPut}(X[m{-}2]),\;
  \ldots,\;
  \textsf{RansEncPut}(X[0]).
\]
We show that for a row-major $R \times C$ array ($m = R \cdot C$),
the identical reverse sequence can be produced without materializing
the full array, by traversing rows from $i = R{-}1$ down to~$0$ and
coefficients from $j = C{-}1$ down to~$0$ within each row:
\begin{align*}
  &\text{for } i = R{-}1 \text{ downto } 0: \\
  &\quad \text{for } j = C{-}1 \text{ downto } 0: \\
  &\qquad \textsf{RansEncPut}\bigl(X[i \cdot C + j]\bigr).
\end{align*}
This allows each coefficient to be fed to the encoder immediately upon
computation, without materializing the full $R \times C$ array.
In Pass~$C_1$~(\stepn{12} in Figure~\ref{fig:sign-stream}), each hint row $\mathbf{h}_i$ ($R{=}k$, $C{=}N$) is
recomputed via row-streaming and immediately encoded;
in Pass~$C_2$~(\stepn{13}), $\textsf{HB}^{z_1}(\mathbf{z}_{1,i})$
($R{=}\ell$, $C{=}N$) is handled analogously.
Both full staging buffers are eliminated entirely.

The encoded output is accumulated in a compact stack-local encoder
context and copied to the signature buffer only after encoding
completes, following an \emph{internal encoding with copy-out} boundary
model.
To verify correctness, we performed differential testing for each
parameter set (\hae{}-2/3/5, 10\,000 signatures each):
for each randomly generated signature, we extracted $\mathbf{h}$ and
$\textsf{HB}^{z_1}(\mathbf{z}_1)$ from the produced signature,
re-encoded them with both the conventional full-array encoder
(\texttt{encode\_h} / \texttt{encode\_hb\_z1}) and the streaming
reverse-order encoder, and compared the resulting byte streams.
In all cases the two outputs were byte-identical.

\paragraph{Incremental transcript hashing.}
Both signing (Pass~A, steps~\stepn{5}--\stepn{6} in Figure~\ref{fig:sign-stream})
and verification (\stepn{8}--\stepn{9} in Figure~\ref{fig:verify-stream})
derive the challenge polynomial by
hashing a transcript that includes the packed high bits of the matrix--vector
product.
The reference implementation allocates a contiguous buffer of
\texttt{POLYVECK\_HIGHBITS\_PACKEDBYTES\,+\,POLYC\_PACKEDBYTES} bytes
(608--1\,184\,B depending on the security level), packs all $k$ rows into
it, and then feeds it to SHAKE256 in one shot.
Because SHAKE256 is a sponge and satisfies
$\mathrm{absorb}(A\|B) = \mathrm{absorb}(A);\,\mathrm{absorb}(B)$,
we replace this monolithic buffer with \emph{row-by-row incremental absorb}:
each row's high bits are packed into a small single-row buffer
(\texttt{POLY\_HIGHBITS\_PACKEDBYTES}), immediately absorbed into the
SHAKE256 state, and discarded before the next row is processed.
After the row loop, the remaining fields ($\mathbf{w}'$ and the message
digest $\mu$) are absorbed and the state is finalized.
This eliminates the full transcript buffer entirely, yielding measurable
stack savings (up to \SI{656}{\byte} for HAETAE-5).

\paragraph{Sparse challenge multiplication with fused accumulation.}
In Pass~B~(\stepn{9} in Figure~\ref{fig:sign-stream}),
for \hae{}-2/3, where the challenge
$c(X) = \sum_{t=1}^{\tau} X^{i_t}$ has Hamming weight~$\tau{=}60$, we
replace NTT-based multiplication by a purely coefficient-domain \emph{signed
shift-and-add} rule in the negacyclic ring
$R_q = \mathbb{Z}_q[X]/(X^N{+}1)$:
\[
  (c \star s)[j] \;=\; \sum_{t=1}^{\tau} \varepsilon_t \cdot s[(j-i_t)\bmod N],
  \qquad
  \varepsilon_t = \begin{cases} +1 & j \geq i_t, \\ -1 & j < i_t, \end{cases}
\]
where $\varepsilon_t$ arises from the negacyclic relation $X^N = -1$.
For \hae{}-5 ($\tau{=}128$, dense binary challenge), we retain
NTT-based multiplication as the dense weight makes a coefficient-domain
approach less efficient.
In all cases we fuse the accumulation directly into the
response: each signed shift (or partial product) is added in place to the
corresponding component of~$\mathbf{y}$, producing~$\mathbf{z}$ without
allocating a separate product polynomial (saving \SI{1024}{\byte} per
multiplication).

\paragraph{Algorithm-level BSS elimination.}
Our complete signing path uses no mutable BSS (\texttt{.bss}\,$=$\,0\,B).
The hyperball sampler's static temporary buffers
(\texttt{tmp\_samples[N+1]} and \texttt{tmp\_signs[(N+7)/8]},
totaling \SI{2088}{\byte} in the reference implementation)
are replaced by a fully streaming Gaussian backend that
generates samples on-the-fly from the SHAKE sponge state, as detailed in
Section~\ref{sec:sampler}.
Our implementation also provides an alternative build configuration
(\texttt{BRS\_BSS\_WORKSPACE}) that places large signing vectors
($\mathbf{y}$, $\mathbf{z}$, $\mathbf{Ay}$) in a static BSS
workspace for reduced recomputation, trading BSS for lower-latency
signing.
In the full-streaming configuration (Table~\ref{tab:pqm4-eval}),
the pass decomposition eliminates this workspace entirely, achieving
\texttt{.data}\,$=$\,0 and \texttt{.bss}\,$=$\,0.

\begin{figure}[t]
    \centering
    \resizebox{1.0\textwidth}{!}{
        \begin{tabular}{c p{0.03in} c p{0.03in} c p{0.03in} c p{0.03in} p{0.2in} p{0.03in} l}
            $\mid \xleftrightarrow{\;\; 32+32 \;\;}\mid$ & &
            $\mid \xleftrightarrow{\;\; 512 \;\;}\mid$ & &
            $\mid \xleftrightarrow{\;\; 1024 \;\;}\mid$ & &
            $\mid \xleftrightarrow{\;\; 1024 \;\;}\mid$ & & & & \\
            \cline{1-1}
            \multicolumn{1}{|c|}{$\mathbf{w}', \rho$} & & & & & & & & & &
            \stepn{0} parse $\sigma$; extract $\rho$ from $\pk$ \\[2pt]
            \multicolumn{1}{|c|}{} & & & & & & & & & &
            {\scriptsize\textbf{--- Pre-pass: $\mathbf{w}'$ and $\|\mathbf{z}_1\|^2$ ---}} \\[2pt]
            \cline{3-3}
            \multicolumn{1}{|c|}{} & &
            \multicolumn{1}{|c|}{$\mathrm{hb\_col}$} & & & & & &
            \tikzmark{vpp2}{} & \tikzmark{vpp3}{} &
            \stepn{1} $\mathrm{hb\_col} \gets \textsf{rANS\_decode}(\sigma, \ell)$ \\[2pt]
            \multicolumn{1}{|c|}{} & &
            \multicolumn{1}{|c|}{} & & & & & &
            \tikzmark{vpp1}{} & \tikzmark{vpp0}{} &
            \stepn{2} $z_{1,\ell} \gets 2^8 \cdot \mathrm{hb\_col} + \mathrm{lb}[\ell]$ \\[1pt]
            \multicolumn{1}{|c|}{} & &
            \multicolumn{1}{|c|}{} & & & & & & & &
            {\scriptsize $\triangleright$ $\|\mathbf{z}_1\|^2 \mathrel{+}= \|z_{1,\ell}\|^2$; $\mathbf{w}' \gets \textsf{LSB}(z_{1,0}-c)$} \\[2pt]
            \cline{3-3}
            \multicolumn{1}{|c|}{} & & & & & & & & & &
            {\scriptsize\textbf{--- Row loop: $r = 0 \ldots K{-}1$ ---}} \\[2pt]
            \cline{3-3} \cline{5-5}
            \multicolumn{1}{|c|}{} & &
            \multicolumn{1}{|c|}{$\mathrm{hb\_col}$} & &
            \multicolumn{1}{|c|}{$z_{1,\ell}$} & & & &
            \tikzmark{vrc2}{} & \tikzmark{vrc3}{} &
            \stepn{3} $\mathrm{hb\_col} \gets \textsf{rANS\_decode}(\sigma, \ell)$ \\[1pt]
            \multicolumn{1}{|c|}{} & &
            \multicolumn{1}{|c|}{} & &
            \multicolumn{1}{|c|}{} & & & & & &
            {\scriptsize $\triangleright$ re-decode per row ($K\times$ recomputation)} \\[2pt]
            \cdashline{5-5}[1pt/1pt]
            \multicolumn{1}{|c|}{} & &
            \multicolumn{1}{|c|}{} & &
            \multicolumn{1}{|c|}{$\widehat{z}_{1,\ell}$} & & & & & &
            \stepn{4} $\widehat{z}_{1,\ell} \gets \textsf{NTT}(z_{1,\ell})$ \\[2pt]
            \cdashline{7-7}[1pt/1pt]
            \multicolumn{1}{|c|}{} & &
            \multicolumn{1}{|c|}{} & &
            \multicolumn{1}{|c|}{} & &
            \multicolumn{1}{|c|}{$t_r$} & &
            \tikzmark{vrc1}{} & \tikzmark{vrc0}{} &
            \stepn{5} $t_r \mathrel{+}= \widehat{\Apo}_{r,\ell} \circ \widehat{z}_{1,\ell}$ \\[1pt]
            \multicolumn{1}{|c|}{} & &
            \multicolumn{1}{|c|}{} & &
            \multicolumn{1}{|c|}{} & &
            \multicolumn{1}{|c|}{} & & & &
            {\scriptsize $\triangleright$ streamed $\widehat{\Apo}$} \\[2pt]
            \cline{3-3} \cdashline{5-5}[1pt/1pt] \cdashline{7-7}[1pt/1pt]
            \multicolumn{1}{|c|}{} & & & &
            \multicolumn{1}{|c|}{} & &
            \multicolumn{1}{|c|}{$t_r$} & & & &
            \stepn{6} $t_r \gets \textsf{NTT}^{-1}(t_r)$ \\[1pt]
            \multicolumn{1}{|c|}{} & & & &
            \multicolumn{1}{|c|}{} & &
            \multicolumn{1}{|c|}{} & & & &
            {\scriptsize $\triangleright$ $\textsf{fromCRT}$ lift to mod $2q$} \\[2pt]
            \cline{3-3}
            \multicolumn{1}{|c|}{} & &
            \multicolumn{1}{|c|}{$\mathrm{h\_row}$} & &
            \multicolumn{1}{|c|}{} & &
            \multicolumn{1}{|c|}{} & & & &
            \stepn{7} $\mathrm{h\_row} \gets \textsf{rANS\_decode}(\sigma, r)$ \\[1pt]
            \multicolumn{1}{|c|}{} & &
            \multicolumn{1}{|c|}{} & &
            \multicolumn{1}{|c|}{} & &
            \multicolumn{1}{|c|}{} & & & &
            {\scriptsize $\triangleright$ reuses hb\_col memory (union)} \\[2pt]
            \cline{5-5}
            \multicolumn{1}{|c|}{} & &
            \multicolumn{1}{|c|}{} & &
            \multicolumn{1}{|c|}{$\mathbf{w}'_{1,r}$} & &
            \multicolumn{1}{|c|}{} & & & &
            \stepn{8} $\mathbf{w}'_{1,r} \gets \textsf{HB}^h(\mathbf{w}'_r) + \mathrm{h\_row}$ \\[2pt]
            \multicolumn{1}{|c|}{} & &
            \multicolumn{1}{|c|}{} & &
            \multicolumn{1}{|c|}{} & &
            \multicolumn{1}{|c|}{} & & & &
            \stepn{9} absorb $\textsf{pack}(\mathbf{w}'_{1,r})$ into $\mathcal{H}$ \\[1pt]
            \multicolumn{1}{|c|}{} & &
            \multicolumn{1}{|c|}{} & &
            \multicolumn{1}{|c|}{} & &
            \multicolumn{1}{|c|}{} & & & &
            {\scriptsize $\triangleright$ incremental transcript hashing} \\[2pt]
            \cline{5-5} \cdashline{7-7}[1pt/1pt]
            \multicolumn{1}{|c|}{} & &
            \multicolumn{1}{|c|}{} & & & &
            \multicolumn{1}{|c|}{$z_{2,r}$} & & & &
            \stepn{10} $z_{2,r} \gets (\alpha \mathbf{w}'_{1,r} - t_r + w'_r)/2$ \\[1pt]
            \multicolumn{1}{|c|}{} & &
            \multicolumn{1}{|c|}{} & & & &
            \multicolumn{1}{|c|}{} & & & &
            {\scriptsize $\triangleright$ $\|\mathbf{z}_2\|^2 \mathrel{+}= \|z_{2,r}\|^2$} \\[2pt]
            \cline{3-3} \cline{7-7}
            \multicolumn{1}{|c|}{} & & & & & & & &
            \tikzmark{vrr1}{} & \tikzmark{vrr0}{} &
            \stepn{11} abort if $\|(\mathbf{z}_1,\mathbf{z}_2)\|_2 > B'$ \\[3pt]
            \multicolumn{1}{|c|}{} & & & & & & & & & &
            \stepn{12} $c' \gets \textsf{SampleBinaryChallenge}_\tau(H(\mathbf{w}'_1, \mathbf{w}', \mu))$ \\[2pt]
            \cline{1-1}
            & & & & & & & & & &
            \stepn{13} return $(c' \stackrel{?}{=} c)$ \\
        \end{tabular}
        \begin{tikzpicture}[remember picture, overlay, >=stealth, shift={(0,0)}, thick]
            \draw[rounded corners,->] ([yshift=\tikzoffset, xshift=5pt] vpp0.east) -- ([yshift=\tikzoffset, xshift=5pt] vpp1.east) -- ([yshift=\tikzoffset, xshift=5pt] vpp2.east) -- ([yshift=\tikzoffset, xshift=5pt] vpp3.east) node [midway, above, sloped, xshift=-4pt, yshift=-2pt] () {\ ${\scriptstyle 0\le \ell < L}$};
            \draw[rounded corners,->] ([yshift=\tikzoffset, xshift=5pt] vrc0.east) -- ([yshift=\tikzoffset, xshift=5pt] vrc1.east) -- ([yshift=\tikzoffset, xshift=5pt] vrc2.east) -- ([yshift=\tikzoffset, xshift=5pt] vrc3.east) node [midway, above, sloped, xshift=-4pt, yshift=-2pt] () {\ ${\scriptstyle 0\le \ell < L}$};
            \draw[rounded corners,->] ([yshift=\tikzoffset, xshift=5pt] vrr0.east) -- ([yshift=\tikzoffset] vrr1.east) -- ([yshift=\tikzoffset] vrc2.east) node [midway, right, rotate=-90, anchor=south, yshift=-2pt] () {\ ${\scriptstyle 0\le r < K}$} -- ([yshift=\tikzoffset] vrc3.east);
        \end{tikzpicture}}
    \caption{
        Memory allocation of the row-streamed \hae{} verification.
        Four memory slots are used: persistent small buffers
        ($\mathbf{w}'$ packed bits + seed $\rho$, 64\,B, retained through
        the final challenge recomputation),
        a \texttt{union\{int8\_t hb\_col[N]; uint16\_t h\_row[N]\}} (512\,B)
        that alternates between decoded high-bits columns in the inner
        $\ell$ loop and the decoded hint row afterwards,
        and two polynomials \texttt{poly z} and \texttt{poly trow}
        (1\,024\,B each).
        The row loop re-decodes $\textsf{HB}^{z_1}(\mathbf{z}_1)$ via
        streaming rANS for each row $r$, trading a factor-$k$
        re-computation for the elimination of a full \texttt{polyveck}
        buffer.
        The incremental SHAKE256 state $\mathcal{H}$ absorbs each row's
        packed high bits immediately.
    }
    \label{fig:verify-stream}
\end{figure}

\subsection{Low-stack verification via view-style decoding and row streaming}
\label{sec:impl:lowstack-verify}

HAETAE verification (Algorithm~\ref{alg:verify} in
Appendix~\ref{sec:app:verify}) reconstructs
\(\tilde{\mathbf w}'_1 = \tilde{\mathbf h} + \textsf{HB}^h(\tilde{\mathbf w}')\) from
\(\tilde{\mathbf w} = \widehat{\Apo}\circ \textsf{NTT}(\tilde{\mathbf z}_1)\) (up to the CRT lift mod~\(2q\)), derives the
auxiliary component \(\tilde{\mathbf z}_2\), and re-hashes the packed transcript
\((\tilde{\mathbf w}'_1, w')\) to recompute the challenge.
Rather than following the column-streamed approach
of~\cite{cryptoeprint:2026/442}, our verification is built on the~\cite{DBLP:conf/africacrypt/BosRS22} method of
trading re-computation for transient memory, combining
(i)~\emph{view-style decoding} with union overlays,
(ii)~\emph{row-streamed} matrix--vector multiplication, and
(iii)~\emph{incremental transcript hashing} that eliminates the
full-transcript buffer (the same technique introduced for signing in
Section~\ref{sec:impl:lowstack-sign}).
This design supports all three security levels, including
HAETAE-2/3 ($D{>}0$), which is not addressed
by~\cite{cryptoeprint:2026/442}, and achieves a uniform verification
stack of \SI{3800}{\byte} across all security levels, a 39\,\% reduction compared
to the \SI{6220}{\byte} reported by~\cite{cryptoeprint:2026/442} for
HAETAE-5 alone.
Figure~\ref{fig:verify-stream} illustrates the resulting memory layout.

\paragraph{View-style decoding.}
We avoid materializing the decoded signature vectors
(steps~\stepn{0}--\stepn{2} in Figure~\ref{fig:verify-stream}).
Concretely, we keep $\textsf{LB}^{z_1}(\lfloor\mathbf{z}_1\rceil)$ as
a pointer into the signature buffer, decode only
$\textsf{HB}^{z_1}(\lfloor\mathbf{z}_1\rceil)$ into an \texttt{int8}
array~(\stepn{1}), and decode $\mathbf{h}$ into a compact \texttt{uint16} array~(\stepn{7}).
These two decoded arrays share a \texttt{union} overlay
(\texttt{int8\_t hb\_col[N]} / \texttt{uint16\_t h\_row[N]}), exploiting
lifetime disjointness: the high-bits columns are consumed during
the inner column loop, after which the same memory is reused for the
decoded hint row.
We also represent $w' \in \mathcal{R}_2$ as a packed bitstring of
$N/8$ bytes. A pre-pass~(\stepn{1}--\stepn{2}) over coefficients recomposes
$\tilde{z}_{1,\ell} = \textsf{HB}^{z_1} \cdot 256 + \textsf{LB}^{z_1}$
on-the-fly, accumulates $\lVert\tilde{\mathbf{z}}_1\rVert_2^2$, and derives
$w' = \textsf{LSB}(\tilde{z}_{1,1} - c)$.
Together, the working set for decoded data is a single 512-byte union plus
a 32-byte packed bitstring, independent of the module dimensions $k$ and $\ell$.
Signatures that fail format checks (length, zero-padding, or rANS
decoder errors) are rejected before any transcript recomputation.

\paragraph{Row-streamed multiplication.}
To minimize peak stack, we compute $\tilde{\mathbf{w}}$ one row at a time
rather than column-by-column
(steps~\stepn{3}--\stepn{11} in Figure~\ref{fig:verify-stream}).
For each output row $r$, we recompute an NTT-domain accumulator
$\tilde{w}_r \gets \sum_{\ell=0}^{L-1} \widehat{\Apo}[r,\ell]\circ \textsf{NTT}(\tilde{z}_{1,\ell})$~(\stepn{3}--\stepn{5}) by
(re)composing and transforming one column $\tilde{z}_{1,\ell}$ at a time and streaming
$\widehat{\Apo}[r,\ell]$ from the public seed $\rho$.
After an inverse NTT and CRT lift via $\textsf{fromCRT}$~(\stepn{6}),
we decode the hint row $\tilde{\mathbf{h}}_r$~(\stepn{7}) and form
$\tilde{w}'_{1,r} = \tilde{\mathbf{h}}_r + \textsf{HB}^h(\tilde{w}'_r)$~(\stepn{8}),
then immediately absorb the packed high bits into the
incremental SHAKE256 transcript state~(\stepn{9}).
Within the same loop we compute $\tilde{z}_{2,r}$~(\stepn{10}), accumulate
$\lVert\tilde{z}_{2,r}\rVert_2^2$, and abort early as soon as the bound test fails~(\stepn{11}).
This reduces the working set to a constant number of polynomials plus small
decoded arrays, eliminating concurrent \texttt{polyveck} temporaries.
The tradeoff is a factor-$k$ increase in re-computation:
each row re-decodes the $\ell$ columns of $\textsf{HB}^{z_1}(\lfloor\mathbf{z}_1\rceil)$
via streaming rANS and re-applies $\ell$ forward NTTs.

In contrast, \cite{cryptoeprint:2026/442}~uses a \emph{column-streamed}
approach that accumulates $\tilde{\mathbf{w}}$ into a full \texttt{polyveck} buffer
($k \times 1\,024$\,B) while sweeping the columns of $\tilde{\mathbf{z}}_1$ once.
Our row-streamed design trades this $k$-fold re-decode cost for a
$(k{-}1) \times 1\,024$\,B reduction in peak live memory.
Combined with the incremental transcript hashing described in
Section~\ref{sec:impl:lowstack-sign}, which absorbs each row's high bits
directly into the SHAKE256 state, the full
\texttt{POLYVECK\_HIGHBITS\_PACKEDBYTES} transcript buffer
(up to 1\,184\,B) is eliminated entirely, an optimization not applied
by either the reference implementation or~\cite{cryptoeprint:2026/442}.

For $D{>}0$ (HAETAE-2/3), the first column of $\widehat{\Apo}$ is derived
on-the-fly from the public seed and the packed public key, adding one
polynomial expansion per row but no additional persistent storage.

\begin{algorithm}[t]
  \caption{Streamed hyperball sampler}
  \label{alg:two-pass-hyperball}  
  \small
  \begin{algorithmic}[1]
    \Require Gaussian dimension $N$, integers $L, K$;
      seed $\mathsf{seed} \in \{0,1\}^{\mathsf{CRHBYTES}}$;
      nonce $\mathsf{nonce} \in \{0,1\}^{16}$;
      hyperball bound $B_0$ and scaling constant $\Lambda$ (e.g.\ $\Lambda = L+K$)
    \Ensure $(\mathbf y_1, \mathbf y_2) \in \mathbb Z^{L \times N} \times \mathbb Z^{K \times N}$
      with $\lVert (\mathbf y_1,\mathbf y_2)\rVert_2 \leq B_0 \Lambda$

    \Repeat
      \State \Comment{Pass 1: compute the (unnormalized) squared norm $S$}
      \State $\InitStream(\mathsf{st}, \mathsf{seed}, \mathsf{nonce})$
      \State $S \gets 0$ \Comment{$S$ is a scalar accumulator (e.g.\ 64-bit integer)}
      \For{$i \gets 0$ to $L+K-1$}
        \For{$j \gets 0$ to $N-1$}
          \State $x_{i,j} \gets \SampleGauss(\mathsf{st})$
          \State $S \gets S + x_{i,j}^2$
        \EndFor
      \EndFor
      \State \Comment{Compute scaling factor $\alpha \approx \dfrac{B_0 \Lambda}{\sqrt{S}}$}
      \State $\alpha \gets \InvSqrt(S, B_0, \Lambda)$

      \State \Comment{Pass 2: regenerate samples and scale immediately}
      \State $\InitStream(\mathsf{st}, \mathsf{seed}, \mathsf{nonce})$
      \For{$i \gets 0$ to $L-1$}
        \For{$j \gets 0$ to $N-1$}
          \State $x_{i,j} \gets \SampleGauss(\mathsf{st})$
          \State $y_{1,i,j} \gets \round(\alpha \cdot x_{i,j})$
        \EndFor
      \EndFor
      \For{$i \gets 0$ to $K-1$}
        \For{$j \gets 0$ to $N-1$}
          \State $x_{L+i,j} \gets \SampleGauss(\mathsf{st})$
          \State $y_{2,i,j} \gets \round(\alpha \cdot x_{L+i,j})$
        \EndFor
      \EndFor
    \Until{$\lVert (\mathbf y_1,\mathbf y_2)\rVert_2^2 \leq B_0^2 \Lambda^2$}
  \end{algorithmic}
\end{algorithm}

In the second pass, the sampler restarts the pseudorandom stream from
the same $(\mathsf{seed}, \mathsf{nonce})$ pair, so that
$\SampleGauss(\mathsf{st})$ reproduces the identical sequence
$(x_{i,j})$. This time, each sample is immediately rescaled and rounded,
\[
  y_{1,i,j} = \round(\alpha x_{i,j})\quad (0 \leq i < \ell),
  \qquad
  y_{2,i,j} = \round(\alpha x_{\ell+i,j})\quad (0 \leq i < k),
\]
and written to the output polynomials $\mathbf y_1$ and $\mathbf y_2$.
The algorithm then evaluates the squared norm
$\bigl\|(\mathbf y_1,\mathbf y_2)\bigr\|_2^2$ and rejects if it exceeds
$B_0^2 \Lambda^2$, repeating the entire two-pass procedure until acceptance.

\subsection{A memory friendly sampler}
\label{sec:sampler}
The hyperball sampler is invoked in Pass~A of signing
(\stepn{1} in Figure~\ref{fig:sign-stream}) to draw the ephemeral
vectors $(\mathbf{y}_1, \mathbf{y}_2)$ used in each rejection-sampling
iteration.
We require a sampler that outputs a random vector
\[
  (\mathbf y_1, \mathbf y_2) \in \mathbb Z^{\ell \times N} \times \mathbb Z^{k \times N},
\]
whose distribution is (approximately) Gaussian, conditioned on the norm
constraint
\[
  \bigl\|(\mathbf y_1,\mathbf y_2)\bigr\|_2 \leq B_0 \Lambda,
\]
for some bound $B_0 > 0$ and scaling factor $\Lambda$ (in our case $\Lambda = \ell+k$).
A naïve implementation samples all $(\ell+k)N$ Gaussian coefficients, stores them in
memory, computes their squared norm, rescales the entire vector, and finally applies
a rejection test. This requires $\Theta((\ell+k)N)$ words of transient storage.
The two-pass approach that avoids this large buffer was first described
by~\cite{lee2024generalized} for BLISS and is also adopted
by~\cite{cryptoeprint:2026/442} for HAETAE-5.
We apply the same high-level structure but further eliminate all
static (BSS) scratch buffers through a fully streaming Gaussian backend.

Algorithm~\ref{alg:two-pass-hyperball} implements what we call the
\emph{two-pass} sampler with \emph{streaming} Gaussian sampling and it 
only uses $O(1)$ additional memory. The sampler is parameterized by a seed
$\mathsf{seed} \in \{0,1\}^{\mathsf{CRHBYTES}}$ and a nonce
$\mathsf{nonce}$, which are used to initialize a pseudorandom stream
$\mathsf{st}$ (e.g. based on SHAKE or \textsf{stream256}).

In the first pass, the algorithm draws $(\ell+k)N$ independent
discrete Gaussian samples $(x_{i,j})_{0 \leq i < \ell+k,\ 0 \leq j < N}$ from
$\SampleGauss(\mathsf{st})$, and maintains only the scalar accumulator
\[
  S \;=\; \sum_{i=0}^{\ell+k-1} \sum_{j=0}^{N-1} x_{i,j}^2.
\]
No sample is stored beyond its contribution to $S$, so the memory footprint of
this pass is independent of $N$, $\ell$, and $k$. After the loop, the algorithm
computes a scaling factor
\[
  \alpha \approx \frac{B_0 \Lambda}{\sqrt{S}},
\]
using a fixed-point approximation of the reciprocal square root.
The following proposition proves that
the two-pass streaming implementation produces an identical output
distribution to the one-pass reference procedure using the same
fixed-point \InvSqrt{} routine.

\begin{proposition}
	Conditioned on acceptance, the joint distribution of
	$(\mathbf{y}_1, \mathbf{y}_2)$ produced by
	Algorithm~\ref{alg:two-pass-hyperball} is identical to that of the
	one-pass reference procedure that (i) draws all $(L+K)N$ Gaussian
	samples into a single array, (ii) computes the scaling factor
	$\alpha = \InvSqrt(\sum x_{i,j}^2, B_0, \Lambda)$, (iii) maps each
	$x_{i,j} \mapsto \round(\alpha x_{i,j})$, and (iv) accepts if and only if
	$\lVert(\mathbf{y}_1,\mathbf{y}_2)\rVert_2^2 \le B_0^2\Lambda^2$.
\end{proposition}

\begin{proof}
	Fix any $\mathsf{seed}$ and $\mathsf{nonce}$.
	Since $\InitStream(\mathsf{st},\mathsf{seed},\mathsf{nonce})$ is deterministic
	and $\SampleGauss(\mathsf{st})$ is a pure function of the stream state,
	restarting the stream with the same $(\mathsf{seed},\mathsf{nonce})$
	reproduces the identical sequence $(x_{i,j})$ in both passes.
	Consequently, the scalar $S = \sum_{i,j} x_{i,j}^2$ and the scaling
	factor $\alpha = \InvSqrt(S, B_0, \Lambda)$ computed in Pass~1 are
	identical to those that would be computed from the full array in the
	one-pass variant.
	The output coefficients $\round(\alpha x_{i,j})$ and the acceptance
	predicate $\lVert(\mathbf{y}_1,\mathbf{y}_2)\rVert_2^2 \le B_0^2\Lambda^2$
	are therefore also identical.
	The two procedures thus define the same mapping from
	$(\mathsf{seed},\mathsf{nonce})$ to $(\mathbf{y}_1,\mathbf{y}_2)$
	(or to rejection), and hence the same conditional distribution given
	acceptance.
\end{proof}

Therefore the only difference is implementation strategy: the two-pass sampler
never stores more than a bounded number of Gaussian samples at any one time.
A further quantitative analysis of the fixed-point approximation error is left to future work.

\paragraph{Streaming Gaussian backend.}
The reference \hae{} sampler allocates temporary arrays
\texttt{samples[N$\cdot$($\ell$+k)]} and \texttt{signs[N$\cdot$($\ell$+k)/8]}
in stack or BSS, requiring $\Theta((\ell+k)N)$ words of transient storage.
Our implementation replaces this with a fully streaming backend:
each discrete Gaussian sample is generated from the SHAKE sponge state,
consumed, and discarded immediately.
Since squeezing one sample at a time produces the identical output as
batch squeezing, this has no effect on the output distribution.
This eliminates all static sampler buffers, reducing the scheme-wide
\texttt{.bss} to zero bytes.



%% file: evaluation.tex
\section{Implementation and Evaluation}
\label{sec:impl:results}

This section evaluates our low-stack \hae{} implementation.
We report primary results on the \texttt{pqm4} benchmarking
framework~\cite{PQM4} (Nucleo-L4R5ZI, ARM Cortex-M4), which enables
direct comparison with the prior work of~\cite{cryptoeprint:2026/442}
and with ML-DSA.
We additionally validate portability under RIOT-OS on two further
targets (nRF52840 and ESP32-C6).

\subsection{Build configuration}
\label{sec:impl:build-config}

Our implementation is controlled by three compile-time knobs:
\begin{center}
\texttt{SAMPLER=2}\quad \texttt{STREAM\_MATRIX=1}\quad \texttt{FROZEN\_A=2}
\end{center}
\texttt{SAMPLER=2} selects the two-pass streaming hyperball sampler
(Section~\ref{sec:sampler});
\texttt{STREAM\_MATRIX=1} regenerates matrix entries on demand from
the public seed; and
\texttt{FROZEN\_A=2} fuses rejection sampling with pointwise
multiplication to eliminate temporary polynomial buffers.
When all three are enabled, the full-streaming signing path
(\texttt{FULL\_STREAM\_SIGN}) is activated automatically, applying
the \emph{Rejection-aware pass decomposition} and \emph{Reverse-order streaming entropy coding} techniques
described in Section~\ref{sec:impl:lowstack-sign}.


\begin{table*}[t]
\centering
\caption{Performance comparison on Nucleo-L4R5ZI (pqm4).
Cycles: \texttt{-O3}, \SI{20}{\mega\hertz}/0\,WS, median of 100 runs ($k = 10^3$).
Stack/size: \texttt{-Os}. Ratio relative to \textbf{ref} (${<}1$: faster).
(C)\,=\,pure C; (asm)\,=\,assembly.
\cite{cryptoeprint:2026/442}: published numbers (HAETAE-5 only).
\textbf{ours\,(C)}/\textbf{ours\,(asm)}: this work.
ML-DSA from \texttt{pqm4}; \textbf{m4fstack} followed \cite{DBLP:conf/africacrypt/BosRS22} strategy.}\label{tab:pqm4-eval}
\footnotesize
\setlength{\tabcolsep}{3pt}
\renewcommand{\arraystretch}{1.10}
\begin{tabular}{l r r r r r r r r r r r}
\toprule
 & \multicolumn{3}{c}{Key generation} & \multicolumn{3}{c}{Signing} & \multicolumn{3}{c}{Verification} & & \\
\cmidrule(lr){2-4}\cmidrule(lr){5-7}\cmidrule(lr){8-10}
Impl. & Cycles & Ratio & Stack & Cycles & Ratio & Stack & Cycles & Ratio & Stack & .text & .total \\
\midrule
\multicolumn{12}{l}{\textbf{HAETAE-2}} \\
\quad ref\,(C)   & \num{9253}\,k   & $1.00\times$ & \num{23844} & \num{43710}\,k  & $1.00\times$ & \num{73112} & \num{1732}\,k  & $1.00\times$ & \num{33448} & \num{26490} & \num{29098} \\
\quad m4f\,(asm)   & \num{6980}\,k   & $0.75\times$ & \num{19772} & \num{18616}\,k  & $0.43\times$ & \num{55684} & \num{998}\,k   & $0.58\times$ & \num{23296} & \num{30104} & \num{30624} \\
\quad \textbf{ours\,(C)}    & \num{11630}\,k  & $1.26\times$ & \num{5848}  & \num{115534}\,k & $2.64\times$ & \num{5968}  & \num{1417}\,k  & $0.82\times$ & \num{4936}  & \num{30416} & \num{30416} \\
\quad \textbf{ours\,(asm)}  & \num{11518}\,k  & $1.24\times$ & \num{5816}  & \num{81416}\,k  & $1.86\times$ & \num{5968}  & \num{1071}\,k  & $0.62\times$ & \num{4824}  & \num{38100} & \num{38100} \\
\midrule
\multicolumn{12}{l}{\textbf{HAETAE-3}} \\
\quad ref\,(C)   & \num{12007}\,k  & $1.00\times$ & \num{41236} & \num{53546}\,k  & $1.00\times$ & \num{113328}& \num{3170}\,k  & $1.00\times$ & \num{54232} & \num{26436} & \num{29204} \\
\quad m4f\,(asm)   & \num{13584}\,k  & $1.13\times$ & \num{29484} & \num{28275}\,k  & $0.53\times$ & \num{83436} & \num{1909}\,k  & $0.60\times$ & \num{31776} & \num{30078} & \num{30758} \\
\quad \textbf{ours\,(C)}    & \num{20998}\,k  & $1.75\times$ & \num{5848}  & \num{181792}\,k & $3.40\times$ & \num{6152}  & \num{2897}\,k  & $0.91\times$ & \num{4840}  & \num{30780} & \num{30780} \\
\quad \textbf{ours\,(asm)}  & \num{21889}\,k  & $1.82\times$ & \num{5920}  & \num{138771}\,k & $2.59\times$ & \num{6152}  & \num{2134}\,k  & $0.67\times$ & \num{4840}  & \num{38464} & \num{38464} \\
\midrule
\multicolumn{12}{l}{\textbf{HAETAE-5}} \\
\quad ref\,(C)   & \num{16491}\,k  & $1.00\times$ & \num{52500} & \num{65896}\,k  & $1.00\times$ & \num{144408}& \num{3968}\,k  & $1.00\times$ & \num{68856} & \num{25658} & \num{28786} \\
\quad m4f\,(asm)   & \num{20115}\,k  & $1.22\times$ & \num{34196} & \num{34995}\,k  & $0.53\times$ & \num{103988}& \num{2532}\,k  & $0.64\times$ & \num{37292} & \num{29756} & \num{30796} \\
\quad \cite{cryptoeprint:2026/442}\,(C) & \num{15328}\,k & $0.93\times$ & \num{5212} & \num{300881}\,k & $4.57\times$ & \num{8092} & \num{4187}\,k & $1.06\times$ & \num{6220} & \num{26494} & \num{27534} \\
\quad \textbf{ours\,(C)}    & \num{11561}\,k  & $0.70\times$ & \num{4816}  & \num{218957}\,k & $3.32\times$ & \num{6136}  & \num{3510}\,k  & $0.88\times$ & \num{4840}  & \num{30150} & \num{30150} \\
\quad \textbf{ours\,(asm)}  & \num{10423}\,k  & $0.63\times$ & \num{4888}  & \num{163125}\,k & $2.48\times$ & \num{6136}  & \num{2736}\,k  & $0.69\times$ & \num{4952}  & \num{37834} & \num{37834} \\
\midrule
\multicolumn{12}{l}{\textbf{ML-DSA-44}} \\
\quad m4f\,(asm)       & \num{1418}\,k  & $1.00\times$  & \num{38312} & \num{3017}\,k   & $1.00\times$  & \num{44832} & \num{1496}\,k  & $1.00\times$  & \num{8912}  & \num{19592} & \num{19592} \\
\quad m4fstack\,(asm)  & \num{1797}\,k  & $1.27\times$  & \num{4408}  & \num{10330}\,k  & $3.42\times$  & \num{5064}  & \num{3816}\,k  & $2.55\times$  & \num{2720}  & \num{24844} & \num{24844} \\
\midrule
\multicolumn{12}{l}{\textbf{ML-DSA-65}} \\
\quad m4f\,(asm)       & \num{2526}\,k  & $1.00\times$  & \num{60840} & \num{5963}\,k   & $1.00\times$  & \num{68896} & \num{2533}\,k  & $1.00\times$  & \num{9888}  & \num{19328} & \num{19328} \\
\quad m4fstack\,(asm)  & \num{3422}\,k  & $1.35\times$  & \num{4408}  & \num{21668}\,k  & $3.63\times$  & \num{6608}  & \num{6771}\,k  & $2.67\times$  & \num{2720}  & \num{24120} & \num{24120} \\
\midrule
\multicolumn{12}{l}{\textbf{ML-DSA-87}} \\
\quad m4f\,(asm)       & \num{4267}\,k  & $1.00\times$  & \num{97704} & \num{7145}\,k   & $1.00\times$  & \num{107912}& \num{4381}\,k  & $1.00\times$  & \num{12064} & \num{19500} & \num{19500} \\
\quad m4fstack\,(asm)  & \num{5770}\,k  & $1.35\times$  & \num{4408}  & \num{27150}\,k  & $3.80\times$  & \num{8144}  & \num{11713}\,k & $2.67\times$  & \num{2720}  & \num{24516} & \num{24516} \\
\bottomrule
\end{tabular}
\end{table*}

\subsection{Results on \texttt{pqm4} framework (Nucleo-L4R5ZI)}
\label{sec:impl:results-pqm4}

All measurements are performed on the
\textbf{NUCLEO-L4R5ZI} board (STM32L4R5ZI, ARM Cortex-M4F,
\SI{2}{\mega\byte} Flash, \SI{640}{\kilo\byte} RAM),
the default \texttt{pqm4} target and the platform used
by~\cite{cryptoeprint:2026/442}.
Speed measurements use \texttt{CLOCK\_BENCHMARK}
(\SI{20}{\mega\hertz}, 0\,WS flash) for cycle-accurate results;
stack measurements use \texttt{CLOCK\_FAST} (\SI{120}{\mega\hertz}).
Stack measurements follow the \texttt{pqm4} framework's built-in
stack-benchmarking infrastructure.
Compiler optimization follows the \texttt{pqm4} convention:
\texttt{-Os} for stack and code-size measurements,
\texttt{-O3} for cycle-count benchmarks (separate builds).
We report the median over 100 executions.
All reported stack peaks are empirical measurements obtained with
\texttt{arm-none-eabi-gcc}~11.3.1 under the default \texttt{pqm4}
build configuration (\texttt{-Os}, no link-time optimization).
The \texttt{noinline} pass boundaries rely on compiler support for
the \texttt{\_\_attribute\_\_((noinline))} annotation.

We benchmark four \hae{} configurations:
\textbf{ref\,(C)} (latest C reference code~\cite{HAETAE_Final_spec}),
\textbf{m4f\,(asm)} (pqm4 assembly, older spec~\cite{cheon2024haetae}),
\textbf{ours\,(C)} (this work, pure C), and
\textbf{ours\,(asm)} (this work with assembly NTT and sampler from~\cite{cheon2024haetae}).
Both \textbf{ours} variants use the build knobs from
Section~\ref{sec:impl:build-config}.
Since the implementation of~\cite{cryptoeprint:2026/442} is not publicly
available, we report their published numbers directly.
Table~\ref{tab:pqm4-eval} presents the results.

\paragraph{Stack usage.}
Compared to the reference, \textbf{ours\,(C)} reduces
peak Signing stack by 91.8--95.8\,\% across all levels, bringing it
to \SI{5968}{\byte} (HAETAE-2), \SI{6152}{\byte} (HAETAE-3), and
\SI{6136}{\byte} (HAETAE-5).
Verification stack ranges from \SI{4840}{\byte} to \SI{4936}{\byte},
an 85--93\,\% reduction.
Key generation ranges from \SI{4816}{\byte} (HAETAE-5) to
\SI{5848}{\byte} (HAETAE-2/3).
\textbf{ours\,(asm)} achieves comparable stack across all operations;
the small differences (e.g.\ \SI{4824}{\byte} vs.\ \SI{4936}{\byte}
for HAETAE-2 verification) reflect differing callee frame sizes in the
assembly routines.
Note that the \textbf{m4f} implementation is based on an older version
of the specification and includes ARM assembly optimizations, so a
direct comparison with our implementation is not entirely fair;
we include it for completeness.
For HAETAE-5, the only level reported
by~\cite{cryptoeprint:2026/442}, our implementation achieves lower
stack across all three operations:
\begin{itemize}
\item Key generation: \SI{4816}{\byte} vs.\ \SI{5212}{\byte}
  ($-7.6$\,\%).
  Caller-level union analysis
  (Section~\ref{sec:impl:lowstack-keygen}) packs the FFT workspace
  and sampling scratch into shared memory slots.
\item Signing: \SI{6136}{\byte} vs.\ \SI{8092}{\byte} ($-24$\,\%).
  The pass decomposition with \texttt{noinline} boundaries
  (Section~\ref{sec:impl:lowstack-sign}) ensures the peak is bounded by
  $S_{\mathrm{driver}} + \max(S_A, S_B, S_{C_1}, S_{C_2})$, and the
  \emph{Reverse-order streaming entropy coding} eliminates the full hint and
  $\textsf{HB}^{z_1}$ staging buffers.
\item Verification: \SI{4840}{\byte} vs.\ \SI{6220}{\byte}
  ($-22$\,\%).
  Our row-streamed design
  (Section~\ref{sec:impl:lowstack-verify}) replaces the
  column-streamed \texttt{polyveck} accumulator with a single
  polynomial, combined with view-style decoding and incremental
  transcript hashing.
\end{itemize}
Both variants have \texttt{.data}\,$=$\,0 and
\texttt{.bss}\,$=$\,0, whereas~\cite{cryptoeprint:2026/442} reports
a \SI{1040}{\byte} gap between \texttt{.text} and \texttt{.total},
whose breakdown is not disclosed.

\paragraph{Cycle counts.}
Key generation for HAETAE-5 is 30\,\% faster than the reference
in \textbf{ours\,(C)} ($0.70\times$) and 37\,\% faster with \textbf{ours\,(asm)}
($0.63\times$).
Since $d{=}0$ key generation
(Algorithm~\ref{alg:keygen_d_eq_0}) performs the singular-value norm
rejection before expanding the matrix $\mathbf{A}$,
with an acceptance rate of roughly 10\,\%, most iterations only
execute the lightweight norm check, and the matrix--vector product
runs only once after acceptance.
For HAETAE-2/3 ($d{>}0$), the norm check is interleaved with the
matrix--vector product, so each rejected sample also pays the
streaming matrix cost, resulting in $1.24$--$1.82\times$ slower
key generation.

Signature generation requires $1.86$--$3.40\times$ the cycles of the
reference (\textbf{ours\,(C)}: $2.64$--$3.40\times$; \textbf{ours\,(asm)}: $1.86$--$2.59\times$),
reflecting the cost of seed-based recomputation in every pass.
For HAETAE-5, \textbf{ours\,(C)} achieves \num{218957}\,k cycles,
27\,\% faster than \cite{cryptoeprint:2026/442} (\num{300881}\,k),
with 24\,\% lower stack;
\emph{Component-level early rejection} (Section~\ref{sec:impl:lowstack-sign})
reduces the average cost of rejected iterations by skipping
unnecessary multiply-accumulate operations, contributing to this
speedup.
\textbf{Ours\,(asm)} further reduces this to \num{163125}\,k ($-46$\,\%
vs.\ \cite{cryptoeprint:2026/442}).

Verification is faster than the reference at all levels
for both variants
(\textbf{ours\,(C)}: $0.82$--$0.91\times$; \textbf{ours\,(asm)}: $0.62$--$0.69\times$).
For HAETAE-5, \cite{cryptoeprint:2026/442}~achieves $1.06\times$
with \SI{6220}{\byte} stack,
whereas \textbf{ours\,(C)} achieves $0.88\times$ with
\SI{4840}{\byte}, and \textbf{ours\,(asm)} reaches $0.69\times$.
We attribute the speedup primarily to the smaller working set,
which reduces memory traffic on the Cortex-M4.

\paragraph{Code size.}
\textbf{Ours\,(C)} shows a modest \texttt{.text} increase compared to
the reference (\num{30150} vs.\ \num{28786} for HAETAE-5, $+4.7$\,\%),
reflecting the streaming rANS encoder, fused multiplication, and pass
decomposition logic.
\textbf{Ours\,(asm)} is larger
(\num{37834} for HAETAE-5, $+31$\,\% vs.\ ref)
due to the inlined NTT and CDT sampler routines.
Both variants place all constant tables in read-only memory,
resulting in \texttt{.data}\,$=$\,0 and \texttt{.bss}\,$=$\,0;
the reported \texttt{.text} and \texttt{.total} are therefore
identical.
ML-DSA scheme-only code sizes (\SIrange{19}{24}{\kilo\byte}) are
smaller than \hae{}, reflecting the simpler structure of ML-DSA.

\subsection{Results on RIOT-OS (nRF52840 and ESP32-C6)}
\label{sec:impl:results-riot}

\begin{table*}[t]
\centering
\caption{%
  Low-stack evaluation on RIOT-OS (\textbf{this work}).
  nRF52840: ARM Cortex-M4F @ \SI{64}{\mega\hertz}, \SI{256}{\kilo\byte} RAM\@.
  ESP32-C6: RISC-V RV32IMAC @ \SI{80}{\mega\hertz}, \SI{512}{\kilo\byte} RAM\@.
  Cycles measured via \texttt{xtimer} function from RIOT-OS.}
  \label{tab:haetae-riot-eval}
\small
\setlength{\tabcolsep}{3pt}
\renewcommand{\arraystretch}{1.10}
\begin{tabular}{l r r r r r r}
\toprule
 & \multicolumn{2}{c}{Key generation} & \multicolumn{2}{c}{Signing} & \multicolumn{2}{c}{Verification} \\
\cmidrule(lr){2-3}\cmidrule(lr){4-5}\cmidrule(lr){6-7}
Platform & Cycles & Stack & Cycles & Stack & Cycles & Stack \\
\midrule
\multicolumn{7}{l}{\textbf{HAETAE-2}} \\
\quad nRF52840 &
  \num{29039}\,k & \num{5988} &
  \num{262721}\,k & \num{6096}  &
  \num{2023}\,k  & \num{3984}  \\
\quad ESP32-C6  &
  \num{49371}\,k & \num{6100} &
  \num{476223}\,k & \num{6068}  &
  \num{3619}\,k  & \num{4004} \\
\midrule
\multicolumn{7}{l}{\textbf{HAETAE-3}} \\
\quad nRF52840 &
  \num{46520}\,k & \num{6208} &
  \num{829928}\,k & \num{6288}  &
  \num{3838}\,k  & \num{3984}  \\
\quad ESP32-C6  &
  \num{80811}\,k & \num{6580} &
  \num{1519566}\,k & \num{6260}  &
  \num{7146}\,k  & \num{4004} \\
\midrule
\multicolumn{7}{l}{\textbf{HAETAE-5}} \\
\quad nRF52840 &
  \num{13612}\,k & \num{5096} &
  \num{193913}\,k & \num{6406}  &
  \num{4970}\,k  & \num{3976}  \\
\quad ESP32-C6  &
  \num{21513}\,k & \num{5254} &
  \num{526047}\,k & \num{6244}  &
  \num{9536}\,k  & \num{4004}  \\
\bottomrule
\end{tabular}
\end{table*}

To validate portability beyond the \texttt{pqm4} bare-metal environment, we ran
the same low-stack build under RIOT-OS~\cite{riot:home} on two targets:
the \textbf{Nordic nRF52840} (ARM Cortex-M4F, \SI{64}{\mega\hertz},
\SI{256}{\kilo\byte} RAM) and the \textbf{Espressif ESP32-C6}
(RISC-V RV32IMAC, \SI{80}{\mega\hertz}, \SI{512}{\kilo\byte} RAM).
Stack usage is measured via a high-watermark technique: the unused
stack region below the current stack pointer is painted with a canary
pattern before each operation, and the deepest overwrite is recorded
after the operation returns.
Timing uses RIOT's \texttt{xtimer} layer converted to cycles via
\texttt{coreclock\_hz}.
Results are averaged over 100 iterations.

As Table~\ref{tab:haetae-riot-eval} shows, peak stack figures are
consistent between the two RIOT-OS targets.
For HAETAE-2, signature generation peaks at \SI{6096}{\byte} (nRF52840) and
\SI{6068}{\byte} (ESP32-C6),
confirming that the memory footprint is determined by the algorithm
and build knobs, not by the ISA or OS layer.

\paragraph{Cycle-count gap between Cortex-M4F and RV32IMAC.}
Despite a 25\,\% higher clock frequency, the ESP32-C6 requires
substantially more cycles for every operation.
The slowdown is most pronounced for signing and other operations,
the main reason is the absence of DSP/SIMD hardware on the in-order
RV32IMAC core: the Cortex-M4F benefits from a single-cycle 32-bit
multiplier suited to the NTT-heavy polynomial arithmetic in \hae{}.

\begin{table}[t]
  \centering
  \caption{RAM budget for Verification and Signing (bytes).
  Verify total: $|\textit{sig}| + \text{stack}$ (public key in flash).
  Sign total: $|\textit{sk}| + |\textit{sig}| + \text{stack}$.
  Cycles: $k = 10^3$, from \texttt{pqm4} (Table~\ref{tab:pqm4-eval}).}
  \label{tab:ram-budget}
  \resizebox{\linewidth}{!}{
  \setlength{\tabcolsep}{3.5pt}
  \begin{tabular}{lrrrrrrrrr}
    \toprule
    & & & & \multicolumn{3}{c}{Verification} &
    \multicolumn{3}{c}{Signing} \\
    \cmidrule(lr){5-7}\cmidrule(lr){8-10}
    Scheme & $|\textit{pk}|$ & $|\textit{sk}|$ & $|\textit{sig}|$ &
    Stack & Total & Cycles &
    Stack & Total & Cycles \\
    \midrule
    \hae{}-2 (\textbf{ours\,(C)}) & 992 & 1\,408 & 1\,474 & 4\,936 & 6\,410 & \num{1417}\,k &
    5\,968 & 8\,850 & \num{115534}\,k \\
    \hae{}-3 (\textbf{ours\,(C)}) & 1\,472 & 2\,112 & 2\,349 & 4\,840 & 7\,189 & \num{2897}\,k &
    6\,152 & 10\,613 & \num{181792}\,k \\
    \hae{}-5 (\textbf{ours\,(C)}) & 2\,080 & 2\,752 & 2\,948 & 4\,840 & 7\,788 & \num{3510}\,k &
    6\,136 & 11\,836 & \num{218957}\,k \\
    \midrule
    \hae{}-5~\cite{cryptoeprint:2026/442} & 2\,080 & 2\,752 & 2\,948 & 6\,220 & 9\,168 & \num{4187}\,k &
    8\,092 & 13\,792 & \num{300881}\,k \\
    \midrule
    ML-DSA-44 (m4fstack) & 1\,312 & 2\,560 & 2\,420 & 2\,720 & 5\,140 & \num{3816}\,k &
    5\,064 & 10\,044 & \num{10330}\,k \\
    ML-DSA-65 (m4fstack) & 1\,952 & 4\,032 & 3\,309 & 2\,720 & 6\,029 & \num{6771}\,k &
    6\,608 & 13\,949 & \num{21668}\,k \\
    ML-DSA-87 (m4fstack) & 2\,592 & 4\,896 & 4\,627 & 2\,720 & 7\,347 & \num{11713}\,k &
    8\,144 & 17\,667 & \num{27150}\,k \\
    \bottomrule
  \end{tabular}
  }
\end{table}

%% file: conclusion.tex
\section{Discussion}
\label{sec:discussion}

To provide a fair basis for comparison, the discussion below is based
exclusively on \texttt{pqm4} measurements
(Table~\ref{tab:pqm4-eval}) using the same board, clock
configuration, and measurement infrastructure for all schemes.

\paragraph{Verification on 8\,KB devices.}
On many constrained platforms, the public key is stored in flash memory,
so the RAM budget for verification consists of the received
signature plus the stack required to run the algorithm
(Table~\ref{tab:ram-budget}), as is the case in,
e.g.,~\cite{acns/BanegasZBHS22,cryptoeprint:2023/965}.
On devices with as little as \SI{8}{\kilo\byte} of SRAM, the
available RAM is shared among \texttt{.data}, \texttt{.bss}, and the
runtime stack, meaning that static data sections directly reduce the
stack budget.
Our implementation has \texttt{.data}\,$=$\,0 and
\texttt{.bss}\,$=$\,0, leaving the full RAM available for the stack.
In contrast, \cite{cryptoeprint:2026/442}~reports a
\SI{1040}{\byte} gap between \texttt{.text} and \texttt{.total},
but does not disclose the section-level breakdown.
\textbf{Ours\,(C)} fits all three \hae{} security levels within
\SI{8}{\kilo\byte} (\SI{6410}{\byte} to \SI{7788}{\byte}),
whereas~\cite{cryptoeprint:2026/442} requires \SI{9168}{\byte}
for \hae{}-5 verification alone.
Compared to ML-DSA m4fstack, \hae{} verification
is $2.34$--$3.34\times$ faster at each comparable security level
(\num{1417}\,k vs.\ \num{3816}\,k at level~2,
\num{2897}\,k vs.\ \num{6771}\,k at level~3,
\num{3510}\,k vs.\ \num{11713}\,k at level~5).
This makes \hae{} attractive for signature-verification-centric
deployments such as firmware authentication and IoT attestation.

\paragraph{Signature and key generation on 16\,KB devices.}
For signing, the device must hold the secret key, the output
signature, and the signing working set simultaneously
(Table~\ref{tab:ram-budget}).
All \hae{} levels fit within \SI{16}{\kilo\byte}
(\SI{8850}{\byte} to \SI{11836}{\byte}),
while ML-DSA-87 exceeds this budget (\SI{17667}{\byte}) due to the
large secret key (\SI{4896}{\byte}) and signing stack
(\SI{8144}{\byte}).
At security level~5, \hae{}-5 requires \SI{11836}{\byte} for
signing, roughly \SI{5.7}{\kilo\byte} less than ML-DSA-87.
Key generation is comfortable for both schemes: our implementation
requires at most \SI{5848}{\byte} of stack, well within the
\SI{16}{\kilo\byte} budget.

%% file: appendix.tex
\section{Key Generation}
\label{sec:app:keygen}
\begin{figure}[H]
  \centering
  \begin{minipage}[t]{0.48\textwidth}
    \begin{algorithm}[H]
      \caption{$\text{KeyGen}(1^\lambda)$ for $d > 0$}
      \label{alg:keygen_d_gt_0}
      \footnotesize
      \begin{algorithmic}[1]
        \Ensure $(\pk, \sk)$
        \State $\text{seed} \gets \{0,1\}^{\rho_0}$
        \State $(\text{seed}_\Apo, \text{seed}_\sk, K) \gets H_{\text{gen}}(\text{seed})$
        \State $(\apo_{\text{gen}}, \widehat{\Apo}_{\text{gen}}) \gets \textsf{ExpandA}_d(\text{seed}_\Apo)$
        \State $(\text{counter}_{\sk}, \text{flag}) \gets (0, \text{true})$
        \While{$\text{flag}$}
          \State $(\spo_{\text{gen}}, \epo_{\text{gen}}) \gets \textsf{ExpandS}(\text{seed}_\sk, \text{counter}_\sk)$
          \State $\bpo \gets \apo_{\text{gen}} + \textsf{NTT}^{-1}(\widehat{\Apo}_{\text{gen}} \circ \textsf{NTT}(\spo_{\text{gen}})) + \epo_{\text{gen}} \bmod q$
          \State $(\bpo_0, \bpo_1) \gets (\textsf{LowBits}^{\pk}(\bpo), \textsf{HighBits}^{\pk}(\bpo))$
          \State $(\spo_1, \spo_2) \gets (\spo_{\text{gen}}, \epo_{\text{gen}} - \bpo_0)$
          \State $\text{counter}_{\sk} \gets \text{counter}_{\sk} + 1$
          \If{$\mathcal{N}(\spo_1, \spo_2) \le \gamma^2 n$}
            \State $\text{flag} \gets \text{false}$
          \EndIf
        \EndWhile
        \State $tr \gets H(\text{seed}_\Apo, \bpo_1)$
        \State {\textbf{return} $(\pk = (\text{seed}_\Apo, \bpo_1), \sk = (\spo_1, \spo_2, K, tr, \text{seed}_\Apo, \bpo_1))$}
      \end{algorithmic}
    \end{algorithm}
  \end{minipage}
  \hfill
  \begin{minipage}[t]{0.48\textwidth}
    \begin{algorithm}[H]
      \caption{$\text{KeyGen}(1^\lambda)$ for $d = 0$}
      \label{alg:keygen_d_eq_0}
      \footnotesize
      \begin{algorithmic}[1]
        \Ensure $(\pk, \sk)$
        \State $\text{seed} \gets \{0,1\}^{\rho_0}$
        \State $(\text{seed}_\Apo, \text{seed}_\sk, K) \gets H_{\text{gen}}(\text{seed})$
        \State $(\text{counter}_{\sk}, \text{flag}) \gets (0, \text{true})$
        \While{$\text{flag}$}
          \State $(\spo_{\text{gen}}, \epo_{\text{gen}}) \gets \textsf{ExpandS}(\text{seed}_\sk, \text{counter}_\sk)$
          \State $(\spo_1, \spo_2) \gets (\spo_{\text{gen}}, \epo_{\text{gen}})$
          \State $\text{counter}_{\sk} \gets \text{counter}_{\sk} + 1$
          \If{$\mathcal{N}(\spo_1, \spo_2) \le \gamma^2 n$}
            \State $\text{flag} \gets \text{false}$
          \EndIf
        \EndWhile
        \State $\widehat{\Apo}_{\text{gen}} \gets \textsf{ExpandA}_d(\text{seed}_\Apo)$
        \State $\widehat{\bpo} \gets -2 \left( \widehat{\Apo}_{\text{gen}} \circ \textsf{NTT}(\spo_{\text{gen}}) + \textsf{NTT}(\epo_{\text{gen}}) \right) \bmod q$
        \State $tr \gets H(\text{seed}_\Apo, \widehat{\bpo})$
        \State {\textbf{return} $(\pk = (\text{seed}_\Apo, \widehat{\bpo}), \sk = (\spo_1, \spo_2, K, tr, \text{seed}_\Apo, \widehat{\bpo}))$}
      \end{algorithmic}
    \end{algorithm}
  \end{minipage}
  \caption{Implementation specification of HAETAE key generation~\cite{NISTPQC-ADD-R1:HAETAE23,HAETAE_Final_spec}: $d > 0$ applies to HAETAE-2 and HAETAE-3, while $d = 0$ applies to HAETAE-5.}
  \label{fig:keygen-d-comparison}
\end{figure}

\begin{figure}[H]
    \centering
    \resizebox{0.88\textwidth}{!}{
        \begin{tabular}{c p{0.05in} c p{0.05in} c p{0.05in} p{0.25in} p{0.05in} l}
            $\mid \xleftrightarrow{\;\; 128 \;\;}\mid$ & &
            $\mid \xleftrightarrow{\;\; 1024 \;\;}\mid$ & &
            $\mid \xleftrightarrow{\;\; 2048 \;\;}\mid$ & & & & \\
            \cline{1-1}
            \multicolumn{1}{|c|}{$\rho, \sigma, K$} & & & & & & & &
            \stepn{0} $\rho, \sigma, K := H_{\text{gen}}(\text{seed})$ \\
            \multicolumn{1}{|c|}{} & & & & & & & &
            {\scriptsize\textbf{--- Pass 1: norm check ---}} \\
            \cline{3-3}
            \multicolumn{1}{|c|}{} & &
            \multicolumn{1}{|c|}{$\text{sum}{=}0$} & & & &
            \tikzmark{hfrej2}{} & \tikzmark{hfrej3}{} &
            \stepn{1} $\text{sum} \gets 0$ \\
            \cline{5-5}
            \multicolumn{1}{|c|}{} & &
            \multicolumn{1}{|c|}{} & &
            \multicolumn{1}{|c|}{$\spo_{1,j}$} & &
            \tikzmark{hfp1s2}{} & \tikzmark{hfp1s3}{} &
            \stepn{2} $\spo_{1,j} \gets \textsf{ExpandS}(\sigma, j)$ \\
            \cdashline{5-5}[1pt/1pt]
            \multicolumn{1}{|c|}{} & &
            \multicolumn{1}{|c|}{$\text{}$} & &
            \multicolumn{1}{|c|}{$\textsf{FFT}(\spo_{1,j})$} & &
            \tikzmark{hfp1s1}{} & \tikzmark{hfp1s0}{} &
            \stepn{3} $\text{sum} \mathrel{+}= \|\spo_{1,j}\|^2$ \quad {\scriptsize(in-place FFT)} \\
            \cline{5-5}
            \multicolumn{1}{|c|}{} & &
            \multicolumn{1}{|c|}{} & &
            \multicolumn{1}{|c|}{$\spo_{2,i}$} & &
            \tikzmark{hfp1e2}{} & \tikzmark{hfp1e3}{} &
            \stepn{4} $\spo_{2,i} \gets \textsf{ExpandS}(\sigma, i)$ \\
            \cdashline{5-5}[1pt/1pt]
            \multicolumn{1}{|c|}{} & &
            \multicolumn{1}{|c|}{$\text{}$} & &
            \multicolumn{1}{|c|}{$\textsf{FFT}(\spo_{2,i})$} & &
            \tikzmark{hfp1e1}{} & \tikzmark{hfp1e0}{} &
            \stepn{5} $\text{sum} \mathrel{+}= \|\spo_{2,i}\|^2$ \quad {\scriptsize(in-place FFT)} \\
            \cline{3-3} \cline{5-5}
            \multicolumn{1}{|c|}{} & &
            \multicolumn{1}{|c|}{$\mathcal{N}(\text{sum})$} & & & &
            \tikzmark{hfrej1}{} & \tikzmark{hfrej0}{} &
            \stepn{6} reject if $\mathcal{N}(\text{sum}) > \gamma^2 N$ \\
            \cline{3-3}
            \multicolumn{1}{|c|}{} & &
            \multicolumn{1}{|c|}{$\widehat{\bpo}{=}0$} & & & & & &
            {\scriptsize\textbf{--- Pass 2: compute $\widehat{\bpo}$ ---}} \\
            \cline{5-5}
            \multicolumn{1}{|c|}{} & &
            \multicolumn{1}{|c|}{} & &
            \multicolumn{1}{|c|}{$\spo_{1,j}$} & &
            \tikzmark{hfp2c2}{} & \tikzmark{hfp2c3}{} &
            \stepn{7} $\spo_{1,j} \gets \textsf{ExpandS}(\sigma, j)$ {\scriptsize(re-sample)} \\
            \multicolumn{1}{|c|}{} & &
            \multicolumn{1}{|c|}{} & &
            \multicolumn{1}{|c|}{} & & & &
            \stepn{8} \textbf{write $\spo_{1,j}$ to $\sk$} {\scriptsize(where $i{=}0$)} \\
            \cdashline{5-5}[1pt/1pt]
            \multicolumn{1}{|c|}{} & &
            \multicolumn{1}{|c|}{} & &
            \multicolumn{1}{|c|}{$\widehat{\spo}_{1,j}$} & & & &
            \stepn{9} $\widehat{\spo}_{1,j} \gets \textsf{NTT}(\spo_{1,j})$ \\
            \cdashline{3-3}[1pt/1pt]
            \multicolumn{1}{|c|}{} & &
            \multicolumn{1}{|c|}{$\widehat{\bpo}$} & &
            \multicolumn{1}{|c|}{} & &
            \tikzmark{hfp2c1}{} & \tikzmark{hfp2c0}{} &
            \stepn{10} $\widehat{\bpo} \mathrel{+}= \widehat{\Apo}_{i,j} \circ \widehat{\spo}_{1,j}$ \quad {\scriptsize(streamed $\widehat{\Apo}_{i,j}$)} \\
            \cdashline{3-3}[1pt/1pt] \cline{5-5}
            \multicolumn{1}{|c|}{} & &
            \multicolumn{1}{|c|}{$\widehat{\bpo}$} & & & & & &
            \stepn{11} $\widehat{\bpo} \gets \textsf{fromMont}(\widehat{\bpo})$ \\
            \cline{5-5}
            \multicolumn{1}{|c|}{} & &
            \multicolumn{1}{|c|}{} & &
            \multicolumn{1}{|c|}{$\spo_{2,i}$} & & & &
            \stepn{12} $\spo_{2,i} \gets \textsf{ExpandS}(\sigma, i)$ \\
            \multicolumn{1}{|c|}{} & &
            \multicolumn{1}{|c|}{} & &
            \multicolumn{1}{|c|}{} & & & &
            \stepn{13} \textbf{write $\spo_{2,i}$ to $\sk$} \\
            \cdashline{5-5}[1pt/1pt]
            \multicolumn{1}{|c|}{} & &
            \multicolumn{1}{|c|}{} & &
            \multicolumn{1}{|c|}{$\widehat{\spo}_{2,i}$} & & & &
            \stepn{14} $\widehat{\spo}_{2,i} \gets \textsf{NTT}(\spo_{2,i})$ \\
            \cdashline{3-3}[1pt/1pt]
            \multicolumn{1}{|c|}{} & &
            \multicolumn{1}{|c|}{$\widehat{\bpo}$} & &
            \multicolumn{1}{|c|}{} & & & &
            \stepn{15} $\widehat{\bpo} \gets \widehat{\bpo} + \widehat{\spo}_{2,i}$ \\
            \cdashline{3-3}[1pt/1pt]
            \multicolumn{1}{|c|}{} & &
            \multicolumn{1}{|c|}{$-2\widehat{\bpo}$} & &
            \multicolumn{1}{|c|}{} & & & &
            \stepn{16} $\widehat{\bpo} \gets -2\widehat{\bpo}$ \\
            \cline{3-3} \cline{5-5}
            \multicolumn{1}{|c|}{} & & & & & &
            \tikzmark{hfp2r1}{} & \tikzmark{hfp2r0}{} &
            \stepn{17} \textbf{write $\widehat{\bpo}$ to $\pk$} {\scriptsize(NTT domain)} \\
            \cline{1-1}
            & & & & & & & &
            \stepn{18} $\textsf{pack\_remain\_pk}(\rho), \textsf{pack\_remain\_sk}(K)$ \\
        \end{tabular}
        \begin{tikzpicture}[remember picture, overlay, >=stealth, shift={(0,0)}, thick]
            \draw[rounded corners,->] ([yshift=\tikzoffset, xshift=5pt] hfp1s0.east) -- ([yshift=\tikzoffset, xshift=5pt] hfp1s1.east) -- ([yshift=\tikzoffset, xshift=5pt] hfp1s2.east) -- ([yshift=\tikzoffset, xshift=5pt] hfp1s3.east) node [midway, above, sloped, xshift=-4pt, yshift=-2pt] () {\ ${\scriptstyle 0\le j < \ell}$};
            \draw[rounded corners,->] ([yshift=\tikzoffset, xshift=5pt] hfp1e0.east) -- ([yshift=\tikzoffset, xshift=5pt] hfp1e1.east) -- ([yshift=\tikzoffset, xshift=5pt] hfp1e2.east) -- ([yshift=\tikzoffset, xshift=5pt] hfp1e3.east) node [midway, above, sloped, xshift=-4pt, yshift=-2pt] () {\ ${\scriptstyle 0\le i < k}$};
            \draw[rounded corners, dashed, ->] ([yshift=\tikzoffset, xshift=5pt] hfrej0.east) -- ([yshift=\tikzoffset, xshift=-10pt] hfrej1.east) -- ([yshift=\tikzoffset, xshift=-10pt] hfrej2.east) -- ([yshift=\tikzoffset, xshift=5pt] hfrej3.east) node [midway, above, sloped, xshift=2pt] () {\ ${\scriptstyle \text{reject}}$};
            \draw[rounded corners,->] ([yshift=\tikzoffset, xshift=5pt] hfp2c0.east) -- ([yshift=\tikzoffset, xshift=5pt] hfp2c1.east) -- ([yshift=\tikzoffset, xshift=5pt] hfp2c2.east) -- ([yshift=\tikzoffset, xshift=5pt] hfp2c3.east) node [midway, above, sloped, xshift=-4pt, yshift=-2pt] () {\ ${\scriptstyle 0\le j < \ell}$};
            \draw[rounded corners,->] ([yshift=\tikzoffset, xshift=5pt] hfp2r0.east) -- ([yshift=\tikzoffset] hfp2r1.east) -- ([yshift=\tikzoffset] hfp2c2.east) node [midway, right, rotate=-90, anchor=south, yshift=-2pt] () {\ ${\scriptstyle 0\le i < k}$} -- ([yshift=\tikzoffset] hfp2c3.east);
        \end{tikzpicture}}
    \caption{
        Memory allocation of \hae{}-5 Key generation ($d{=}0$).
        Two caller-level unions share memory across passes:
        \texttt{union\{sum[N]; poly b\}} (1\,024\,B) and
        \texttt{union\{fft\_input[FFT\_N]; poly s\}} (2\,048\,B).
        Pass~1 checks the norm before expanding $\widehat{\Apo}$;
        Pass~2 computes $\widehat{\bpo} = \textsf{NTT}(-2\bpo)$.
    }
    \label{fig:haetae5-keygen-stream}
\end{figure}

\section{Signature Generation}
\label{sec:app:sign}
\begin{algorithm}[H]
  \caption{Implementation Specification of HAETAE Signing~\cite{NISTPQC-ADD-R1:HAETAE23,HAETAE_Final_spec}}
  \label{alg:sign}  
  \begin{algorithmic}[1]
    \Require $\sk = (\spo_1, \spo_2, K, tr, \text{seed}_\Apo, \psi)$, message $M$
    \Ensure signature $\sigma$
    \State $\widehat{\Apo} \gets \textsf{UnpackA}_d(\text{seed}_\Apo, \psi)$
    \State $\mu \gets H_{\texttt{gen}}(tr, M)$
    \State $\text{seed}_{ybb} \gets H_{\texttt{gen}}(K, \mu)$
    \State $(\kappa, \sigma) \gets (0, \bot)$
    \While{$\sigma = \bot$}
      \State $(\mathbf{y}_1, \mathbf{y}_2, b, b', \kappa) \gets \textsf{ExpandYbb}(\text{seed}_{ybb}, \kappa)$
      \State $\mathbf{w} \gets \textsf{NTT}^{-1}(\widehat{\Apo} \circ \textsf{NTT}(\lfloor \mathbf{y}_1 \rceil)) + 2 \cdot \lfloor \mathbf{y}_2 \rceil \bmod q$
      \State $\mathbf{w}' \gets \textsf{fromCRT}(\mathbf{w}, \lfloor y_{1,1} \rceil)$
      \State $\mathbf{w}'_1 \gets \textsf{HighBits}^h(\mathbf{w}')$
      \State $\rho \gets H(\mathbf{w}'_1, \textsf{LSB}(\lfloor y_{1,1} \rceil), \mu)$
      \State $c \gets \textsf{SampleBinaryChallenge}_\tau(\rho)$
      \State $\widehat{c} \gets \textsf{NTT}(c)$
      \State $z_{1,1} \gets y_{1,1} + (-1)^b \cdot c$
      \State $(\mathbf{z}_1)_{2..\ell} \gets (\mathbf{y}_1)_{2..\ell} + (-1)^b \textsf{NTT}^{-1}(\widehat{c} \circ \widehat{\spo}_1)$
      \State $\mathbf{z}_2 \gets \mathbf{y}_2 + (-1)^b \textsf{NTT}^{-1}(\widehat{c} \circ \widehat{\spo}_2)$
      \If{$\|(\mathbf{z}_1, \mathbf{z}_2)\|_2 < B'$ \textbf{and} $(\|2(\mathbf{z}_1, \mathbf{z}_2) - (\mathbf{y}_1, \mathbf{y}_2)\|_2 > B$ \textbf{or} $b' = 1)$}
        \State $\mathbf{h} \gets \mathbf{w}'_1 - \textsf{HighBits}^h(\mathbf{w}' - 2\lfloor \mathbf{z}_2 \rceil) \bmod^+ \frac{2(q-1)}{\alpha_h}$
        \State $\sigma \gets \textsf{PackSig}(\textsf{HighBits}^{z_1}(\lfloor \mathbf{z}_1 \rceil), \textsf{LowBits}^{z_1}(\lfloor \mathbf{z}_1 \rceil), \mathbf{h}, c)$
      \EndIf
    \EndWhile
  \end{algorithmic}
\end{algorithm}

\section{Verification}
\label{sec:app:verify}
\begin{algorithm}[H]
  \caption{Implementation Specification of HAETAE Verification~\cite{NISTPQC-ADD-R1:HAETAE23,HAETAE_Final_spec}}
  \label{alg:verify}  
  \begin{algorithmic}[1]
    \Require $\pk = (\text{seed}_\Apo, \psi)$, message $M$, signature $\sigma$
    \Ensure Accept or Reject
    \State $\widehat{\Apo} \gets \textsf{UnpackA}_d(\text{seed}_\Apo, \psi)$
    \State $(\textsf{HighBits}^{z_1}(\lfloor \mathbf{z}_1 \rceil),\, \textsf{LowBits}^{z_1}(\lfloor \mathbf{z}_1 \rceil),\, \mathbf{h},\, c) \gets \textsf{UnpackSig}(\sigma)$ \Comment{Can fail and rejection}
    \State $\tilde{\mathbf{z}}_1 \gets \textsf{HighBits}^{z_1}(\lfloor \mathbf{z}_1 \rceil) \cdot 256 + \textsf{LowBits}^{z_1}(\lfloor \mathbf{z}_1 \rceil)$
    \State $w' \gets \textsf{LSB}(\tilde{z}_{1,1} - c)$
    \State $\tilde{\mathbf{w}} \gets \widehat{\Apo} \circ \textsf{NTT}(\tilde{\mathbf{z}}_1) \bmod q$
    \State $\tilde{\mathbf{w}}' \gets \textsf{fromCRT}(\tilde{\mathbf{w}}, w')$
    \State $\tilde{\mathbf{w}}'_1 \gets \tilde{\mathbf{h}} + \textsf{HighBits}^h(\tilde{\mathbf{w}}') \bmod^+ \frac{2(q-1)}{\alpha_h}$
    \State $\tilde{\mathbf{z}}_2 \gets [\tilde{\mathbf{w}}'_1 \cdot \alpha_h + w'\mathbf{j} - \tilde{\mathbf{w}}' \bmod 2q]\,/\,2$ \Comment{Addition with $w'$ only for first vector element}
    \State $\tilde{\mu} \gets H_{\texttt{gen}}(\text{seed}_\Apo, \psi, M)$
    \State \textbf{return} $(c = \textsf{SampleBinaryChallenge}_\tau(H(\tilde{\mathbf{w}}'_1, w', \tilde{\mu}))) \;\wedge\; (\|(\tilde{\mathbf{z}}_1, \tilde{\mathbf{z}}_2)\| < B'')$
  \end{algorithmic}
\end{algorithm}